\documentclass[11pt,a4paper]{article}

\usepackage[margin=1.05in]{geometry}
\usepackage[T1]{fontenc}
\usepackage[utf8]{inputenc}
\usepackage{lmodern}
\usepackage{microtype}

\usepackage{amsmath,amsthm,amssymb}
\usepackage{mathtools}
\usepackage[authoryear,numbers]{natbib}

\usepackage{algorithm}
\usepackage{algpseudocode}
\usepackage{booktabs}
\usepackage{balance} % for balancing columns on the final 
\usepackage{etoolbox}
\usepackage{cleveref}

\usepackage{tikz}
\usetikzlibrary{arrows,decorations.pathmorphing,decorations.pathreplacing,backgrounds,positioning,fit,matrix}
\usetikzlibrary{shapes,calc}
\usepackage{todonotes}

\usepackage{authblk}

\usepackage{thmtools}
\usepackage{thm-restate}
\newtheorem{theorem}{Theorem}
\newtheorem{lemma}[theorem]{Lemma}

\theoremstyle{definition}
\newtheorem{definition}{Definition}

\newtheorem{observation}{Observation}
\newtheorem{corollary}{Corollary}
\newtheorem{reduction}{Reduction}
\newtheorem{example}{Example}

\usepackage{algorithm}
\algnewcommand\algorithmicinput{\textbf{Input:}}
\algnewcommand\algorithmicoutput{\textbf{Output:}}
\algnewcommand\Input{\item[\algorithmicinput]}
\algnewcommand\Output{\item[\algorithmicoutput]}

\newtheoremstyle{probenv}%
{0pt}% Space above
{0em}% Space below
{\hangindent=\parindent}% Body font
{}% Indent amount
{\itshape}% Theorem head font
{{\normalfont:}}% Punctuation after theorem head
{.5em}% Space after theorem head
{}% Theorem head spec (can be left empty, meaning ‘normal’ )

\theoremstyle{probenv}
\newtheorem{problem}{Problem}%[chapter]
\Crefname{problem}{Problem}{Problems}
\Crefname{problem}{Prob.}{Probs.}
\Crefname{parad}{Paradigm}{Paradigms}
\Crefname{parad}{Parad.}{Parads.}

\usepackage{xargs} % gives access to command 'newcommandx' with predefined entries
\usepackage{multirow}

\newcommandx{\problemdef}[6][3=Input,5=Question]{
	\begingroup
	\par\noindent\nopagebreak[4]
	\begin{problem}\label{prob:#2}\vspace{-0.5em}\colorbox{gray!17!white}{\textsc{#1}}\nopagebreak[4]\end{problem}\nopagebreak[4]
	\par\noindent\hangindent=\parindent\textbf{#3}:  #4\nopagebreak[4]
	\par\noindent\hangindent=\parindent\textbf{#5}:  #6
	\par\medskip
	\endgroup
}

\newcommandx{\problemdefO}[6][3=Input,5=Find]{
	\begingroup
	\par\noindent\nopagebreak[4]
	\begin{problem}\label{prob:#2}\vspace{-0.5em}\colorbox{gray!17!white}{\textsc{#1}}\nopagebreak[4]\end{problem}\nopagebreak[4]
	\par\noindent\hangindent=\parindent\textbf{#3}:  #4\nopagebreak[4]
	\par\noindent\hangindent=\parindent\textbf{#5}:  #6
	\par\medskip
	\endgroup
}

\newcommand{\toappendix}[1]{%
	\gappto{\appendixProofText}{{#1}}
}
\newcommand{\appendixproof}[2]{%
	\gappto{\appendixProofText}{\subsection{Proof of \Cref{#1}}\label{proof:#1}#2}
}

\newcommand{\NN}{\mathbb{N}} % natural numbers
\newcommand{\minEfx}{\textsc{minCost-EFx Allocation}}
\newcommand{\EfxOptim}{\textsc{minCost-EFx Allocation}}
\newcommand{\minEfxOptim}{\textsc{minCostFactor-EFx Allocation}}
\newcommand{\minEfxIndiv}{\textsc{$k$-Cost-EFx Allocation}}
\newcommand{\parti}{\textsc{Partition}}

\definecolor{airforceblue}{rgb}{0.38, 0.51, 0.71}\definecolor{lighterblue}{rgb}{0.7, 0.73, 0.89}\definecolor{darkerblue}{rgb}{0.28, 0.41, 0.61}
\definecolor{aliceblue}{rgb}{0.6, 0.73, 0.89}

\title{Minimizing the Cost of EFx Allocations}

\author{Eva Deltl}

\affil{TU Clausthal, Germany}

\date{} % empty = no date

\begin{document}
\maketitle

\begin{abstract}
Ensuring fairness while limiting costs, such as transportation or storage, is an important challenge in resource allocation, yet most work has focused on cost minimization without fairness or fairness without explicit cost considerations.
We introduce and formally define the \minEfx{} problem, where the objective is to compute an allocation that is envy-free up to any item (EFx) and has minimum cost. We investigate the algorithmic complexity of this problem, proving that it is NP-hard already with two agents. On the positive side, we show that the problem admits a polynomial kernel with respect to the number of items, implying that a core source of intractability lies in the number of items. Building on this, we identify parameter-restricted settings that are tractable, including cases with bounded valuations and a constant number of agents, or a limited number of item types under restricted cost models. Finally, we turn to cost approximation, proving that for any $\rho>1$ the problem is not $\rho$-approximable in polynomial time (unless $P=NP$), while also identifying restricted cost models where costs are agent-specific and independent of the actual items received, which admit better approximation guarantees.
\end{abstract}

% \noindent\textbf{Keywords:} Non-centroid fair clustering; Proportional fairness; Core stability; Max-loss objective.

% If you want a compact highlight:
%\paragraph{Highlights.} Counterexample: no $\alpha$-core under max-loss for any $\alpha<1+\tfrac{1}{9}$ (tight).

\section{Introduction}
Allocating resources is rarely free. In many settings, assigning an item to an agent incurs transportation, storage, or coordination costs, and some assignments may be infeasible altogether. In disaster relief, for example, food, water, and medical supplies must be distributed fairly while keeping transportation costs low. Similar constraints arise in task scheduling, where loads should be balanced without inflating completion time, and in institutions that allocate staff or funding under tight budgets while maintaining perceived legitimacy.

As a result, an allocation that satisfies a strong fairness guarantee in the standard cost-free model may be dominated or even infeasible once costs are taken into account. This motivates our study of fair allocation under cost functions.

We introduce a model that seeks cost-minimal allocations satisfying \emph{envy-freeness up to any item (EFx)}. In standard envy-freeness, no agent prefers another’s allocation over their own. EFx \citep{10.1145/3355902} relaxes this by requiring that if an agent envies another, removing any single item from the envied bundle eliminates the envy. By studying EFx under cost functions, we capture scenarios where fairness must coexist with logistical or budgetary constraints.

%Applications of fair allocation almost always involve explicit costs.  %In disaster relief, for example, food, water, and medical supplies must be distributed fairly while keeping transportation costs low. 
%Similar constraints arise in other domains. Task scheduling must distribute processor loads fairly without inflating overall completion time, and universities or companies must divide staff or funding in a way that appears fair, in order to be considered legitimate, while staying within strict budget limits.  These examples illustrate a common theme: fairness is necessary for acceptance and legitimacy, but efficiency cannot be ignored because real-world allocations are shaped by explicit cost structures. This motivates our study of EFx under cost functions, providing a framework that captures both dimensions.

\subsection{Related work}
%Defining fairness is a complex task. While \emph{envy-freeness (EF)}—where no individual prefers another's allocation to their own—is an ideal, it is often challenging or impossible to attain in many real-world situations. For divisible resources, EF can always be achieved, e.g., through cake-cutting protocols \citep{gamow1958puzzle,Brams_Taylor_1996}, but for indivisible goods EF cannot be guaranteed, such as when a single item must be shared by two agents. \todo{We could shorten this to - "Exact envy-freeness is generally unattainable with indivisible goods, which has motivated weaker relaxations." }
%Therefore, alternative fairness concepts are essential. 
Exact envy-freeness is generally unattainable with indivisible items, which has motivated alternative fairness concepts.
\citet{doi:10.1086/664613} introduced \emph{envy-freeness up to one item (EF1)}, which is satisfied if any envy can be eliminated by removing a single item from another agent’s bundle. 
EF1 allocations are guaranteed to exist and can be computed in polynomial time \citep{Lipton}.
\citet{10.1145/3355902} proposed the stronger criterion \emph{envy-freeness up to any item (EFx)}, requiring that envy disappears regardless of which single item is removed. 
Existence of EFx allocations is unresolved in general, though they have been shown to exist in special cases, such as for three agents \citep{10.1145/3616009,10.1145/3580507.3597799} or identical rankings \citep{doi:10.1137/19M124397X}. 
In \citeyear{doi:10.1137/19M124397X}, \citet{doi:10.1137/19M124397X} presented an algorithm for computing an EFx allocation when all agents have the same valuation function $v(\cdot)$, where the only assumption on $v(\cdot)$ is that it is monotone.
In the case of additive valuations with identical rankings, where each agent ranks the items in the same order based on their individual valuations, the algorithm returns an EFx allocation in $\mathcal{O}(mn^3)$ time.
Other fairness relaxations include \emph{proportionality up to one item (PROP1)} and the \emph{maximin share (MMS)} guarantee \citep{doi:10.1086/664613,caragiannis2019unreasonable}.

Efficiency is a significant consideration in this context, with concepts such as the Nash Social Welfare (NSW), which seeks to maximize the product of agents' utilities \citep{75cce021-3400-34e6-b232-0a306c8c3f69}, providing a balance between fairness and efficiency. Weighted versions of NSW have also been studied \citep{feng_et_al:LIPIcs.ICALP.2024.63}, embedding priorities or implicit costs into the objective. Another key concept is that of Pareto optimality, which ensures that no reallocation can improve the situation of one agent without simultaneously worsening it for another. 

The complexity of efficient and envy-free resource allocation has been studied in the past \citep{10.5555/3060621.3060636,comp}. \citet{AZIZ2023773} have examined the complexity of the existence and computation of allocations that are both welfare-maximizing and fair for various fairness concepts, including EF1 and EFx. %They found that deciding whether an allocation exists that is both utilitarian-maximal (UM), i.e., one that maximizes the sum of all agents' utility, and EFx is NP-complete, even when only two agents are involved. 
Further results characterize EFx in restricted settings \citep{10.1007/978-3-031-43254-5_17,Morozov_Ignatiev_Dementiev_2024}, showing that the number of EFx allocations can be significantly smaller than expected, which narrows the algorithmic search space. 
Approximation approaches have also been developed, particularly for MMS allocations \citep{10.1145/3147173,GARG2021103547} and for NSW maximization, which yields EF1 and MMS guarantees \citep{caragiannis2019unreasonable}.

Several works have explored approximation notions related to EFx, either by donating items to charity \citep{10.1145/3616009}, sometimes to achieve approximate Nash welfare \citep{approx2}, or by designing algorithms that guarantee approximate EFx \citep{doi:10.1137/19M124397X,approx}.

A complementary line of research explicitly analyzes the trade-off between fairness and efficiency. 
\citet{Bei_2021} introduced the notion of the \emph{price of fairness} for indivisible items, quantifying the efficiency loss (e.g. welfare) caused by imposing fairness constraints. 

Our problem also resembles classical challenges in bin packing \citep{10.1007/s00453-024-01224-5} and job scheduling \citep{koutecky_et_al:LIPIcs.ISAAC.2020.18}, which have been extensively studied in approximation algorithms and computational complexity. 
In bin packing, items of varying sizes must be placed into a limited number of bins without exceeding capacity, paralleling our task of distributing items among agents under EFx constraints. 
In job scheduling, tasks are assigned to resources (e.g., processors) to optimize criteria such as minimizing completion time, illustrating the same trade-off between efficiency and fairness; fairness-aware scheduling has also been explored \citep{10.1007/978-3-031-43254-5_14,BILO201665}.  

Beyond these theoretical connections, fairness combined with cost considerations has also been studied in applied optimization. 
In logistics, \citet{REY201819} integrated envy-freeness into food rescue and delivery, coupling fair allocation with vehicle routing costs. 
In facility location, \citet{FILIPPI2021959} proposed a bi-objective model that minimizes infrastructure costs while ensuring equitable access. 

\subsection{Our Contributions}

In order to address the limitations of traditional fair division methods, we introduce a cost function to the concept of EFx, thereby accounting for a variety of factors, including diverse needs, geographic considerations, and resource availability. Consequently, we formally define the following \emph{Optimization Problem}:

\problemdefO{\EfxOptim}{minEFx}
{A set $M = \{x_1, \dots, x_m\}$ of items, a set $A = \{a_1, \dots, a_n\}$ of agents, a valuation function $v: \mathcal{P}(M) \rightarrow \NN$ and agent-specific cost functions $c_i: \mathcal{P}(M) \rightarrow \NN$ for each agent $a_i \in A$.}
{An EFx Allocation $\mathcal{X} = (X_1, \dots, X_n)$, which minimizes the total cost function $c(\mathcal{X}) = \sum_{a_i \in A} c_i(X_i)$.}

Given the uncertainty surrounding the existence of EFx allocations, the problem definition is restricted to scenarios in which all agents share the same valuation function $v$. Under this assumption, \citet{doi:10.1137/19M124397X} proved that EFx allocations are guaranteed to exist.

In this study, we limit our attention to additive valuation functions, which define the value of a set of items as the sum of the individual item values. Our objective is to identify scenarios in which computations can be performed in a more efficient manner.

\smallskip
Next, we provide an overview of our computational hardness and algorithmic results. We begin by showing that the complexity of the general problem is primarily driven by the number of items. Building on this, we establish that \minEfx \ admits a problem kernel of size polynomial in $m$ (\Cref{cor:pkm}). To deepen our understanding of the problem’s structure, we make several key observations before proving that the problem is NP-hard for two or more agents (\Cref{thm:2np}, \Cref{thm:snp}) and $W[1]$-hard with respect to the number of agents $n$ even if all numbers in the input are encoded in unary (\Cref{w1}). 

Motivated by these hardness results, we then investigate parameter-restricted variants. In particular, we show that when valuations are polynomially bounded in the number of items, and the number of agents is constant, the problem can be solved in polynomial time (\Cref{thm:polybound-cost}). At the same time, we establish that, the general problem admits no polynomial-time $\rho$-approximation for any $\rho>1$ (\Cref{thm:inapprox2}).

This lead to the analysis of a structured cost model where agent cost functions are entirely valuation-dependent.
%, allowing stronger guarantees.
 In this setting, we analyze instances with a bounded number of item types and design an $\mathcal{O}(n^2 \cdot m^{3\beta})$ algorithm (\Cref{thm:types}), where $\beta$ is the number of item types. We further initiate the study of approximation algorithms for this variant, proving an inapproximability bound of $\tfrac{4}{3}$ (\Cref{thm:inapprox}). See \Cref{tab:overview} for an overview of our complexity results.

\begin{table}[t]
	\centering
	\setlength{\tabcolsep}{9pt}
	\renewcommand{\arraystretch}{1.2}
	\begin{tabular}{@{}lcc@{}}
		\toprule
		& General costs & Restricted costs \\
		\midrule
		-
		& \multicolumn{2}{c}{weakly NP-hard already for $n=2$} \\
		\addlinespace[3pt]
		By $n$
		& \multicolumn{2}{c}{$W[1]$-hard;\ \ XP if $v_{\max}\le \mathrm{poly}(m)$} \\
		By $m$
		& \multicolumn{2}{c}{polynomial kernel} \\
		By $\beta$
		& -- & XP \\
		\bottomrule
	\end{tabular}
	\caption{Overview of results for \minEfx.
	Here, $n$ is the number of agents, $m$ the number of items, and $\beta$ the number of item types.}
	\label{tab:overview}
\end{table}

Many details, marked by $\star$, can be found in the appendix.
%a full version of this paper.

% ---------- Colors ----------
\definecolor{fptgreen}{RGB}{161,228,118}
\definecolor{npred}{RGB}{224,60,49}
\definecolor{npred!light}{RGB}{244,170,165}
\definecolor{ppkyellow}{RGB}{255,221,120}
\definecolor{unknownblue}{RGB}{190,210,255}
\colorlet{edgegray}{black!55}

%%%%%%%%%%%%%%%%%%%%%%%%%%%%%%%%%%%%%%%%%%%%%%%%%%%%%%%%%%%%%%%%%%%%%%%%

\section{Preliminaries}

We denote by $A=\{a_1,\dots,a_n\}$ the set of \emph{agents} and by $M=\{x_1,\dots,x_m\}$ the set of \emph{items}. Each agent $a_i \in A$ has a \emph{cost function} $c_i:\mathcal{P}(M)\to\mathbb{N}$ and a \emph{valuation function} $v_i:\mathcal{P}(M)\to\mathbb{N}$ defined over the set of items in $M$; for a single item $g\in M$, we write $v_i(x)$ as shorthand for $v_i(\{x\})$. %A valuation function $v$ is said to be \emph{monotone} if $S \subset T$ implies $v(S) \leq v(T)$ for all $S,T \subseteq M$.\todo{nötig?} 
A valuation function $v$ is \emph{additive} if $v(S) = \sum_{x \in S} v({x})$ for all $S \subseteq M$. 
%Additive valuation functions are, by definition, also monotone. 
We say that agents share the \emph{same valuation function} if $v_i = v_j$ for all agents $a_i,a_j \in A$. 
%Agents have \emph{identical rankings} if, for any two items $x$ and $x'$ in $M$, the inequality $v_i({x}) \geq v_i({x'})$ holds for an agent $a_i \in A$ if and only if it holds for every other agent $a_j \in A$.

An \emph{allocation} $\mathcal{X} = (X_1, \dots, X_n)$ represents a partition of a set of indivisible items into $n$ subsets,  which may also be referred to as bundles.

\smallskip \noindent
\emph{Fairness.}
	Given an allocation $\mathcal{X} = (X_1, \dots, X_n)$, we say that an agent $a_i$ \emph{envies} an agent $a_j$ if $v_i(X_i)<v_i(X_j)$.
\begin{itemize}
	\item[\textbf{EFx}](Envy-freeness up to any item) ensures envy can be eliminated by removing any single item from another agent's bundle. Formally, for any two agents \(a_i\) and \(a_j\) with allocations \(X_i\) and \(X_j\), for all items \(x \in X_j\):
	\[
	v_i(X_i) \geq v_i(X_j \setminus \{x\})
	\]
\end{itemize}

We say that an agent $a_i$ \emph{strongly envies} an agent $a_j$ if $v_i(X_i)<v_i(X_j\setminus \{x\})$ for some $x\in X_j$.
We say that an allocation $\mathcal{X}$ is \emph{feasible} with respect to EFx if no agent strongly envies another.

\emph{Parameterized complexity.} 
We use basic notations from parameterized complexity and algorithmics as described by \citet{dblp1894779},  \citet{FELLOWS200953}, \citet{Nie06} and \citet{doi:10.1137/S0097539703427203}.

A \emph{parameterized problem} is a subset \( L \subseteq \Sigma^* \times \mathbb{N} \), where \( \Sigma \) is a finite alphabet. $L$ is considered \emph{fixed-parameter tractable} (FPT) if there exists an algorithm that can determine, if for every given instance $(x,k)$, $k$ denoting the parameter, it can be decided whether $(x,k) \in L$ (\texttt{yes}-instance) or $(x,k) \notin L$ (\texttt{no}-instance) in $f(k) \cdot |x|^{\mathcal{O}(1)}$ time, where $f$ is some computable function, depending solely on the parameter \( k \). The class of such problems is denoted as FPT. The \emph{class XP} contains parameterized problems that can be solved in time \( |x|^{f(k)} \) for every instance \( (x, k) \), where \( f \) is a computable function. A problem is considered \emph{para-NP-hard} with respect to a parameter \( k \) if it remains NP-hard even when \( k \) is fixed to a constant value.

We also consider the concept of \emph{kernelization}. For a problem \( L \) with a parameter \( k \), \emph{kernelization} involves transforming every instance \( I \) into a smaller instance \( I' \) with a new parameter \( k' \), such that the size of \( I' \) and the value of \( k' \) are bounded by some function \( g(k) \). The original instance \( I \) is a yes-instance if and only if the reduced instance \( I' \) is a yes-instance. If the function \( g(k) \) is polynomial in \( k \), then the problem \( L \) is said to admit a \emph{polynomial kernel}.

\emph{Approximation.} We use basic concepts from the study of approximation algorithms as described by \citet{vazirani2001approximation} and \citet{williamson2011design}.

%%%%%%%%%%%%%%%%%%%%%%%%%%%%%%%%%%%%%%%%%%%%%%%%%%%%%%%%%%%%%%%%%%%%%%%%

\section{General Cost Functions}

We begin by showing that \minEfx\ is weakly NP-hard, even in the case of just two agents. This motivates a parameterized study of the problem, where we establish that \minEfx\ admits a polynomial kernel with respect to the number of items~$m$. Specifically, we prove that every instance can be reduced, in polynomial time, to an equivalent instance involving at most $m^3$ agents. Moreover, when valuations are bounded and the number of agents is constant, we obtain a polynomial-time algorithm.

\subsection{NP-hardness}\label{chap:np}

Many proofs presented in this paper build on a fundamental observation about EFx allocations:

\begin{observation}\label{thm:diff}
	In every EFx allocation \( \mathcal{X} = (X_1, \dots, X_n) \) where the agents have identical valuation functions, the difference in valuation between any two sets allocated to different agents is bounded by the value of the most valuable item.
\end{observation}

\begin{observation}[$\star$]\label{threshold}
		In every EFx allocation with $n$ agents, every bundle must have a value of at least
		\[
		\frac{v(M)-\sum_{t=1}^{n-1} v(x_{(t)})}{n} \ ,
		\] where $x_{(1)},x_{(2)},\dots,x_{(m)}$ are the items in $M$ sorted by decreasing valuation.
\end{observation}

\appendixproof{threshold}{
To satisfy the EFx condition, it must hold for all pairs of agents $a_i,a_j \in A$ that
\begin{align*}
	v(X_i) - \min_{x \in X_i} v(x) \leq v(X_j).
\end{align*}
Let $v(M) = \sum_{x \in M} v(x)$ and let $x_{(1)},x_{(2)},\dots,x_{(m)}$ be the items in $M$ sorted by decreasing
valuation. Since for each $i\neq j$ we denote by $\min_{x\in X_i} v(x)$ the value of some item in $X_i$, and the bundles are disjoint, the sum over $(n-1)$ agents is at most the sum of the $(n-1)$ largest item values. Thus, summing
the inequalities over all $a_i\in A$ and taking $v(X_j)=\min_{a_j \in A} v(X_j)$ yields:
\begin{align*}
	n \cdot \min_{a_j \in A} v(X_j)
	&\geq v(M) - \sum_{a_i \in A\setminus\{a_j\}} \min_{x \in X_i} v(x) \\
	&\ge v(M) - \sum_{t=1}^{n-1} v(x_{(t)}).
\end{align*}
Hence,
\[
\min_j v(X_j) \geq \frac{v(M) - \sum_{t=1}^{n-1} v(x_{(t)})}{n}.
\]
Therefore, every bundle in an EFx allocation must have a value of at least
$\frac{v(M)-\sum_{t=1}^{n-1} v(x_{(t)})}{n}$, completing the proof.\qed}

This condition imposes a minimum threshold for each agent's bundle value. By designing instances where this threshold is met only under nearly equal distribution of value, we can leverage hardness reductions from classical number problems such as \textsc{Partition} (\Cref{prob:parti}) and \textsc{Bin Packing} (\Cref{prob:bin}). This leads to our first hardness result.

\begin{theorem}[$\star$]\label{thm:2np} 
	\minEfx \ is weakly NP-hard, even with two agents.
\end{theorem}

In the following proof we give a polynomial-time many-one reduction from the weakly NP-hard problem \parti \ \citep{NP-co1}.

\problemdef{\parti}{parti}
{A multiset $S = \{s_1, \dots, s_m\}$ of $m$ positive integers with $\sum_{s \in S} s = 2\cdot T$.}
{Is there a partition of $S$ into $S_1,S_2$ such that every subset sums up to $T$?}

\begin{proof}[Proof (\Cref{thm:2np})]
	Let $I = S$ be an instance of \parti, where $\sum_{s \in S} s = 2\cdot T$ and $|S| = m$. We construct an instance $I' := (\{a_1, a_2\}, M, v, c)$ of \minEfx \ in polynomial time as follows.
	\vspace{.5em}\\
	\textbf{Construction:}
	For each integer $s_i \in S$ we add an item $x_i$ to $M$ with valuation $v(x_i) = s_i$. Furthermore we add one item $ \{x_{\min}\} $ with valuation $v(x_{\min}) = 0$. The cost functions look as follows: \begin{align*}
		c_1(x_i) &= 0 & \forall x_i \in M \\
		c_2(x_i) &= x_i & \forall x_i \in M \setminus x_{\min} \\
		c_2(x_{\min}) &= T+1
	\end{align*} This finishes the construction. 
	We refer to the appendix for the full proof.
\end{proof}

\appendixproof{thm:2np}{
	\textbf{Correctness:} We claim that $I$ is a \texttt{yes}-instance if and only if the minimum cost of $I'$ is $\leq T$. 
	
	\smallskip
	($\Rightarrow$) Let $S_1, S_2$ be a solution of $I$. We claim that $\mathcal{X} = (X_1, X_2)$ with \begin{align*}
		X_1 = \bigcup_{s_i \in S_1}\{ x_i \} \cup \{x_{\min}\} \text{ and } X_2 = \bigcup_{s_i \in S_2} \{x_i \}
	\end{align*} is a solution to $I'$ of cost $T$. Note that by construction $v(X_1) = v(X_2) = T$. Which means \begin{align*}
		\text{cost}(X) &= c_1(X_1) + c_2(X_2) \\
		&= 0 + v(X_2) \\
		&= T \leq cost(X)
	\end{align*} Lastly we have to show that $\mathcal{X}$ is an EFx Allocation. As $v(X_1) = v(X_2) = T$, $a_1$ and $a_2$ do not envy each other. Hence $\mathcal{X}$ is EFx and has cost $\leq T$.
	
	\smallskip
	($\Leftarrow$) Let $\mathcal{X} = (X_1, X_2)$ be a solution of $I'$. Note that as $c_2(x_{\min}) = T+1 > T$, item $x_{\min}$ has to be assigned to bundle $X_1$. 
	Since $\mathcal{X}$ is a solution of $I'$, we have $T \geq c_2(X_2) \geq v(X_2)$, which as $\sum_{x \in M} v(x) = 2 \cdot T$ implies
	\begin{align*}
	v(X_1) = 2\cdot T - v(X_2) \ge T.
	\end{align*}
	Since $\mathcal{X}$ is an EFx assignment, $v(X_1) - \min_{x \in X_1} v(x) \leq v(X_2)$.
	As $x_{\min}\in X_1$ and $v(x_{\min})=0$, we have $\min_{x \in X_1} v(x)=0$ and thus $v(X_1) \leq v(X_2)$. Together with $v(X_1)+v(X_2)=2T$, it follows that $v(X_1) = v(X_2) = T$. Since $v(x_{\min}) = 0$ it will not influence the sum. We can therefore now define $S_1, S_2$ as follows:
	\begin{align*}
	S_1 = \bigcup_{x_i \in X_1 \setminus \{x_{\min}\}} s_i
	\quad\text{and}\quad
	S_2 = \bigcup_{x_i \in X_2} s_i.
	\end{align*}
	We know that $\sum S_1 + \sum S_2 = \sum S = 2 \cdot T$ and by definition $\sum S_1 = v(X_1) = T = v(X_2) = \sum S_2$. It follows that $\sum  S_1 = \sum S_2$, meaning $I$ is a \texttt{yes}-instance for \parti.
	\smallskip
	
	It follows that \minEfxIndiv \ is weakly NP-hard even if we have only $2$ agents.
}

This computational hardness arises because the costs can be manipulated to force certain assignments, making the problem exceedingly difficult to solve or approximate efficiently (see \Cref{thm:inapprox2}).
Following these hardness results, we now analyze the parameterized complexity of \minEfx.

\subsection{Parameterized Complexity}\label{chap:base}

\begin{theorem}\label{bruteforce}
	Given an instance $I = (A, M, v, c)$ with $n$ agents and $m$ items a solution to \minEfx \ can be computed in $\mathcal{O}(m \cdot n^{2m})$ time.
\end{theorem}

\begin{proof}
	There are at most $n^m$ possible allocations of $m$ items to $n$ agents. 
	Each allocation can be checked for EFx feasibility in $\mathcal{O}(m\cdot n^2)$ time by verifying that every agent values her bundle at least as much as any other bundle with its least-valued item removed. 
	The allocation of minimum cost can then be returned.
	Hence, \minEfx\ can be solved in $\mathcal{O}(m \cdot n^{2m})$ time.
\end{proof}

From \Cref{bruteforce} we can deduce that

\begin{corollary}\label{m}
	\minEfx\ is polynomial-time solvable when the number of items $m$ is constant.
\end{corollary}

We now analyze the case $n \ge m$.

\begin{observation}\label{thm:agents}
	If $n \ge m$ and some item is positively valued by at least $m$ agents, then this item must be allocated alone in every EFx allocation.
\end{observation}

\begin{theorem}\label{nm}
	If $n \geq m$ and every item is positively valued by at least $m$ agents, then \minEfx \ can be solved in polynomial time.
\end{theorem}

\begin{proof} By \Cref{thm:agents}, the EFx constraints enforce a unique feasible partition in which each item is allocated to a distinct agent. 
	To minimize cost, we restrict each item $x_i$ to its $m$ lowest-cost agents $A_i$ with $|A_i|\le m$. 
	Suppose an optimal allocation $\mathcal{X}$ assigns $x_i$ to some agent $b \notin A_i$. 
	Since at most $m$ bundles of $\mathcal{X}$  are nonempty, at least one $a \in A_i$ receives no items.
	As $c_{a}(x_i) \le c_{b}(x_i)$, allocating $x_i$ to $a$ instead of $b$ yields another feasible allocation of no larger cost, and of strictly smaller cost if $c_{a}(x_i) < c_{b}(x_i)$. 
	We can thus restrict our attention to agents in $A_i$.
	
	We now construct a bipartite graph $G=(U,V,E)$ where $U$ is the set of $m$ items, $V$ the $n$ agents, and $(x_i,a)\in E$ has weight $c_{a}(x_i)$ whenever $a \in A_i$. 
	Agents without incident edges can be discarded. 
	A feasible allocation corresponds exactly to a minimum-cost matching of size $m$ in $G$. 
	Since each item contributes at most $m$ edges, $|E| \le m^2$ and $|U|+|V| = \mathcal{O}(m^2)$. 

A minimum-cost matching in a bipartite graph can be reduced to a minimum-cost flow instance~\citep{flow}. Using recent near-linear-time algorithms for min-cost flow, this can be solved in time $m^{1+o(1)}$ \citep{DBLP:journals/jacm/ChenKLPGS25}, where $m$ denotes the number of edges. Hence, when $n\ge m$, the problem is solvable in polynomial time.
\end{proof}

We next consider the general case.

\begin{reduction}\label{red:m}
	If an item is positively valued by more than $m$ agents, retain only its $m$ lowest-cost agents.
\end{reduction}

\begin{reduction}\label{red:m2}
	For every pair of items, if the set of agents who value at least one of them has size fewer than $m$, retain all these agents (say $m'<m$), and additionally retain the $m$ lowest-cost agents for this pair.
\end{reduction}

Applying \Cref{red:m} and \Cref{red:m2} yields the following.

\begin{lemma}\label{lem:m}
	Each instance $(A,M,v,c)$ of \minEfx\ can be reduced in polynomial time to an equivalent instance $(A',M,v,c)$ with at most $m^3$ agents.
\end{lemma}
%\todo{this holds for more general valuations - for identical valuations sets are either valuable or they are not (for all agents). But as EFx not actually known in this setting, maybe just add note with identical ranking?}
\begin{proof}
	Let $\mathcal{X}$ be a cost-optimal EFx allocation for the original instance.  
	Assume that some non-empty bundle $X'$ is assigned to an agent $b \notin A'$.
	
	\medskip\noindent
	\emph{Case 1.} Suppose that $b$ does not strongly envy any bundle when assigned nothing.  
	Then there exists an agent $a \in A'$, kept in the reduction for some item of $X'$, who must have received an empty bundle.
	Reassigning $X'$ to $a$ yields another feasible allocation, since no agent strongly envied $b$,	and the total cost does not increase.
	
	\medskip\noindent
	\emph{Case 2.} Suppose instead that $b$ would strongly envy some bundle if allocated nothing.  
	Then there exists a set $X^*$ of size at least two that $b$ values.  
	By \Cref{red:m} and \Cref{red:m2}, for each pair of items $x,y \in X^*$ we retained the $m$ cheapest agents who value it.  
	Since $b \notin A'$, all these $m$ positions are filled by other agents.  
	Thus, across the items of $X^*$ we retained at least $m$ distinct agents.  
	By the pigeonhole principle, at least one of these retained agents received an empty bundle in $\mathcal{X}$.  
	That agent would then envy $b$’s allocation $X'$, contradicting EFx feasibility.
	
	Therefore we can restrict our attention to agents in $A'$.  
	
	\medskip
	Finally, the size bound follows from the construction. Each single item contributes at most $m$ retained agents and each of the $\binom{m}{2}$ pairs contributes at most $m$ additional agents.
	%, so we save a total of at most $m^2 + m^3$ agents.  
	Hence $|A'| = O(m^3)$, and the reduction is computable in polynomial time.
\end{proof}

Note that this reduction also holds when agents have differing valuations, as long as they are additive. 
Furthermore applying the \citet{FT87} weight reduction twice yields an equivalent instance whose weights are bounded by $poly(m)$ (we refer to the appendix for a full proof), hence the kernel size is polynomial in the number of items.

\appendixproof{cor:pkm}{

We show how we can bound the encoding size of valuations and costs without changing feasibility or optimality.

\begin{theorem}[\citet{FT87}]\label{thm:FT}
	Let $w\in\mathbb{Z}^m$ and $K\in\mathbb{N}$. There is a strongly polynomial algorithm that outputs $w'\in\mathbb{Z}^m$ with $\|w'\|_\infty \le poly(m,\log K)$ such that for all $x\in\mathbb{Z}^m$ with $\|x\|_1\le K$,
	\[
	\mathrm{sign}\,\langle w,x\rangle \;=\; \mathrm{sign}\,\langle w',x\rangle.
	\]
\end{theorem}

We apply \Cref{thm:FT} to the common valuation vector and to the agents’ cost vectors.

\emph{Valuations.} Since valuations are identical and additive, there exists a vector 
$\mathbf{v} \in \mathbb{Z}^m$ such that for every bundle
$S \subseteq M$, we have $v(S) = \langle \mathbf{v}, \chi_S \rangle$, 
where $\chi_S \in \{0,1\}^m$ is the characteristic vector of $S$.
EFx feasibility uses only comparisons of the form
\[
\langle \mathbf{v},\chi_{X_i}\rangle \;\ge\; \langle \mathbf{v},\chi_{X_j}-\mathbf{e}_g\rangle
\quad\text{for } a_i,a_j\in A,\; g\in X_j,
\]
and (trivially) nonnegativity on empty bundles. Every vector on the right-hand side is an integer vector with $\ell_1$-norm at most $m$.
Apply \Cref{thm:FT} to each $\mathbf{v}$ with $K:=m$ to obtain $\mathbf{v}'$.
Then for all $S,T\subseteq M$ with $|S|,|T|\le m$ and all $g\in T$,
\[
\mathrm{sign}\,\big(\langle \mathbf{v},\chi_S\rangle-\langle \mathbf{v},\chi_{T}-\mathbf{e}_g\rangle\big)
\;=\;
\mathrm{sign}\,\big(\langle \mathbf{v}',\chi_S\rangle-\langle \mathbf{v}',\chi_{T}-\mathbf{e}_g\rangle\big),
\]
so every EFx inequality is preserved agentwise. Hence an allocation is EFx under $\{\mathbf{v}\}$ if and only if it is EFx under $\{\mathbf{v}'\}$.

\emph{Costs.}
Each agent $a_i$ has an additive cost vector $\mathbf{c}_i \in \mathbb{Z}^m$,
and the total cost of an allocation $X$ is 
$\sum_{a_i\in A} \langle \mathbf{c}_i, \chi_{X_i}\rangle$. Feasibility in the decision variant of \minEfx\ is of the form
$\sum_{a_i\in A}\langle \mathbf{c}_i,\chi_{X_i}\rangle \le k$,
which can be encoded as
\[
\Big\langle (\mathbf{c}_1,\ldots,\mathbf{c}_{|A|},-k),\;
(\chi_{X_1},\ldots,\chi_{X_{|A|}},1)\Big\rangle \le 0.
\]
whose $\ell_1$-norm is at most $m\cdot |A|+1$.
Apply \Cref{thm:FT} to each $(\mathbf{c}_i,-k_i)$ with 
$K:=m+1$ (or to their concatenation when considering the global sum) to obtain
compressed vectors $(\mathbf{c}_i',-k_i')$; feasibility is preserved.

If we compare two feasible allocations $X,Y$, the sign of
\[
\sum_{a_i \in A} \langle \mathbf{c}_i,\chi_{X_i}\rangle
-\sum_{a_i \in A} \langle \mathbf{c}_i,\chi_{Y_i}\rangle
=\sum_{a_i \in A} \langle \mathbf{c}_i,\chi_{X_i}-\chi_{Y_i}\rangle
\]
must be preserved.  
Since $\|\chi_{X_i}-\chi_{Y_i}\|_1\le 2m$ for each $a_i$, apply \Cref{thm:FT} to each $\mathbf{c}_i$
with $K:=2m$ to obtain $\mathbf{c}_i'$ ensuring
\[
\mathrm{sign}\big\langle \mathbf{c}_i,\chi_{X_i}-\chi_{Y_i}\big\rangle
=
\mathrm{sign}\big\langle \mathbf{c}_i',\chi_{X_i}-\chi_{Y_i}\big\rangle,
\]
and hence total cost comparisons are unchanged.

Thus, replacing $\mathbf{v}, \{\mathbf{c}_i\}, k$ by 
$\mathbf{v}', \{\mathbf{c}_i'\}, k'$ yields an equivalent instance 
whose numerical values are bounded by $poly(m)$, 
without affecting EFx or feasibility. 
Hence the kernel size is polynomial in the number of items.
}

\begin{corollary}[$\star$]\label{cor:pkm}
	\minEfx\ admits a kernel of size polynomial in $m$.
\end{corollary}

Next we prove that when the item valuations are bounded, the \minEfx \ problem can be solved in polynomial time for a fixed number of agents. This finding is particularly relevant to real-world scenarios like divorce settlements or inheritance disputes, where the number of agents is often small. It also complements our later result on the $W[1]$-hardness of the problem with respect to the number of agents (see \Cref{w1}).

\begin{theorem}\label{thm:polybound-cost}
	 \minEfx \ can be solved in $\mathcal{O}(n(m+1)\cdot (v(M)+1)^n)$ time.
\end{theorem}

First, we construct a dynamic programming table $D'$ that will serve as an oracle for our main algorithm.

\begin{definition}
	The \emph{secondary DP table} of an instance $I := (A, M, v, c)$ with $m = |M|$ and $n=|A| $ \ is a Boolean table $D'$ with entries
	\begin{align*}
		D'[i,\mathbf{y}] \in \{ \top, \bot\}
	\end{align*}
	where $i \in \{0 \dots, m\}$ and $\mathbf{y} \in \{0,\dots, v(M)\}^n$, with $\mathbf{y} = (y_1,\dots, y_n)$.
	We interpret the table as follows. An entry $D'[i, \mathbf{y}]$ in the table is true if and only if there is an assignment $(X_1, \dots, X_n)$ of the first $i$ items (ordered by \emph{ascending} valuations) such that the value of agent bundle $v(X_j)$ is at least $y_j$ for all $a_j \in A$.
\end{definition}

\begin{lemma}[$\star$]\label{thm:dp2poly}
	Given an instance $I = (A, M, v, c)$ of \minEfx \ the secondary dynamic programming table $D'$ of $I$ can be computed in $\mathcal{O}(n(m+1)\cdot (v(M)+1)^n)$ time.
\end{lemma}

\appendixproof{thm:dp2poly}{
	We show how to fill the table and compute the entries efficiently.
	
	\bigskip\noindent
	\textbf{Construction:} We start by initializing the table as follows:
	\begin{align*}
		D'[i,\mathbf{y}] = \begin{cases}
			\top, & \text{if } \mathbf{y} = (0, \dots, 0) \text{ and } i = 0\\
			\bot, & \text{otherwise}
		\end{cases}
	\end{align*}
	We arrange the items in order of increasing value.
	For $i > 0$, we calculate all entries $D'[i,\cdot]$ as follows, assuming that all entries $D'[i-1,\cdot]$ have already been computed. We begin with the largest entries, setting $D'[i, \mathbf{y}] $ to \texttt{true} if there exists a $D'[i-1, \mathbf{y'}] $ that is true and meets one of the following conditions: (1) \(\mathbf{y} = \mathbf{y'}\) or (2) \( y'_j = y_j - v(x_i) \) for some \( j \in [n] \), while \( y'_{j'} = y_{j'} \) for all other \( j' \neq j \), meaning that we adjust the value of one component of the vector by subtracting the value of item \( x_i \).

	If \( D'[i,\mathbf{y}] \) is set to \texttt{true}, then also set \( D'[i,\mathbf{z}] = \top \) for all \(\mathbf{z}\) where \( \mathbf{z} \leq \mathbf{y} \). Here, \( \mathbf{z} \leq \mathbf{y} \) means that every entry \( z_j \) in \(\mathbf{z}\) satisfies \( z_j \leq y_j \) in the corresponding entry of \(\mathbf{y}\). Each entry \( y_j \) in the vector \(\mathbf{y}\) represents the sum allocated to \( a_j \) considering items up to \( i \). Since \( z_j \leq y_j \), \( D'[i,\mathbf{z}] \) is set to \texttt{true} because there exists an allocation up to item \( x_i \) where \( a_j \) has a value of at least \( z_j \).
	
	\smallskip\noindent
	\textbf{Running time:} In each row we compute a maximum of $(v(M)+1)^n$ entries. Consequently, checking if there exists a $D'[i-1,\mathbf{y'}]$ set to \texttt{true} that adds up to $\mathbf{y}$ can be done in $\mathcal{O}((v(M)+1)^n)$ time. Since we do this for every item in $M$, we can fill the whole table in $\mathcal{O}((m+1)\cdot (v(M)+1)^n)$.
	
	\smallskip\noindent \textbf{Correctness:}
	We claim that each entry $D'[i, \mathbf{y}]$ is set to \texttt{true} if and only if there exists an assignment $(X_1, \dots, X_n)$ such that the value of an agent bundle $v(X_j)$ considering items up to $i$ sums up to at least $y_j$ for all $j \in [n]$. We prove our claim via induction.
	
	\smallskip
	
	Initially note that for $i=0$, all entries are correct, as before any item is assigned, every bundle has a value of $0$. Therefore, only the value of $D'[0, \mathbf{0}]$ can be achieved. We now assume that there exists an item $x_i \in M$ such that $D'[i,\mathbf{y}]$ is set to \texttt{true} if and only if there exists an assignment $(X_1, \dots, X_n)$, such that the value of the bundle  $X_j$, when considering all items up to $x_i$, sums up to $y_j$ for all $j \in [n]$. We want to show that this also applies when including item $x_{i+1}$. 

\smallskip
	($\Rightarrow$) If \( D'[i+1, \mathbf{y}] \) is set to \texttt{true} there are two possibilities. 
	\begin{itemize}
		\item[ ] 
		\item \( D'[i, \mathbf{y}] \) was set to \texttt{true}, implying that there exists an allocation $X$ of the first $i+1$ items, such that the value of $X_j$ sums up to at least $y_j$ for all $j \in [n]$. The statement stays correct without including $x_{i+1}$, so there is an assignment $(X_1, \dots, X_n)$ such that the value of bundle $X_j$ considering items up to $i$ sums up to $y_j$ for all $j \in [n]$.
		\item There exists a \( \mathbf{y}' \) such that \( D'[i, \mathbf{y}'] \) is \texttt{true}, with $y'_j = y_j - v(x_{i+1})$ for some $j\in [n]$ and $y'_{j'} = y_{j'}$ for all other $j' \neq j$. In that case, as there exists an allocation $X$ of the first $i$ items, such that the value of $X_{j'}$ sums up to at least $y'_{j'}$ for all ${j'} \in [n]$, we can achieve the sum \( \mathbf{y} \) by including \( x_{i+1} \) in $X_j$. 
	\end{itemize} In both cases there exists an assignment $(X_1, \dots, X_n)$ such that the value of bundle $X_j$ considering items up to $x_{i+1}$ sums up to $y_j$ for all $j \in [n]$.
	
	\smallskip($\Leftarrow$) Let us assume the existence of an assignment \( (X_1, \dots, X_n) \) of the first $i+1$ items, such that the value of $X_j$ sums up to $y_j$ for all $j \in [n]$. Let $X_l$ be the set containing item $x_{i+1}$. According to our construction of the dynamic programming table, \( D'[i, \mathbf{y}'] \), with $y'_j = y_j$ for all $j \neq l$ and $y'_l = y_l - v(x_{i+1})$ will be marked \texttt{true} if such an assignment exists. As there exists a \( \mathbf{y}' \) such that \( D'[i, \mathbf{y}'] \) is true and \( y_l = y'_l + v(x_{i+1}) \), our algorithm would find this entry and \( D'[i+1, \mathbf{y}] \) would be set to \texttt{true}.
	
	\bigskip \noindent We have shown that if $D'[i, \mathbf{y}]$ is true for some $i \in \{0, \dots, m\}$, then there exists an assignment up to item $i$ with value vector $\mathbf{y}$ and this property extends to $D'[i+1, \mathbf{y}]$.

$\therefore$ By the principle of induction, the claim holds for all $x_i\in M$. \qed
}

Now, we can define the main dynamic programming table needed for our algorithm.

\begin{definition}
	For an instance $I := (A, M, v, c)$ with $m = |M|$ and $n = |A|$, 
	the \emph{primary DP table} $D$ is an Integer table with entries
	\[
	D[i,\mathbf{y}] \in \mathbb{N}_0 \cup \{\infty\},
	\]
	where $i \in \{0,\dots,m\}$ and $\mathbf{y} \in \{0,\dots,v(M)\}^n$.  
	We interpret $D[i,\mathbf y]$ as the minimum total cost of a partial allocation 
	of the first $i$ items (ordered by \emph{descending} value), with corresponding value vector $\mathbf y$, 
	that is extendable to an EFx allocation with the remaining $m-i$ items.
\end{definition}

\begin{lemma}\label{thm:dppoly-cost}
	Given an instance $I = (A, M, v, c)$ of \minEfx, the dynamic programming table $D$ can be computed in $\mathcal{O}(n(m+1) \cdot (v(M)+1)^n)$ time.
\end{lemma}

\begin{proof}
	We show how to fill the table and compute the entries efficiently.
	
	\smallskip
	\textbf{Construction:} First, we initialize the table as follows:
	\begin{align*}
		D[i,\mathbf{y}] = \begin{cases}
			0, & \text{if } \mathbf{y} = (0, \dots, 0) \text{ and } i = 0\\
			\infty, & \text{otherwise}
		\end{cases}
	\end{align*}
	We arrange the items in descending order of valuation.
	For $i > 0$, we compute all entries $D[i,\cdot]$ as follows, assuming all entries $D[i-1,\cdot]$ have already been computed.
	 For each agent $a_j \in A$ and value vector $\mathbf y$ with $D[i-1,\mathbf y] \neq \infty$, consider the value vector $\mathbf{y'}$ with \( y'_j = y_j + v(x_i) \) and  \( y'_{j'} = y_{j'} \) for all \( j' \neq j \). Let \( \mathbf{z} = (z_1, \dots, z_n) \) be a value vector, with $z_{j'} = \max(y_j - y_{j'}, 0)$ for all $j' \in [n]$. 
	If $D'[m-i,\mathbf z] = \top$, then
	\[
	D[i,\mathbf y'] = \min\{ D[i,\mathbf y'],\, D[i-1,\mathbf y] + c_j(x_i)\}.
	\] 
	
	If these conditions do not hold, meaning there is no way to achieve an EFx allocation with the partial allocation \( \mathbf{y'} \), then we do not update $D[i,\mathbf{y'}]$.
	
	\smallskip
	\textbf{Running time:}
	The table has $(m+1) \cdot (v(M)+1)^n$ entries, each processed in $n$ time.  
	With $D'$ constructed once, the total running time is $\mathcal{O}(n (m+1) \cdot (v(M)+1)^n)$.
	
	\smallskip
	\textbf{Correctness:}  We prove by induction on $i$ that $D[i,\mathbf y]$ equals the minimum total cost of a partial allocation of the first $i$ items (ordered by decreasing value) with value vector $\mathbf y$ that can be extended to an EFx allocation using the remaining $m-i$ items.

	\emph{Base case:} $i=0$. Only the empty allocation of cost $0$ is feasible, as encoded.  
	
	\emph{Induction step:} Assume correctness for $i-1$. 
	Consider a transition that assigns $x_i$ to agent $a_j$.
	Since items are processed in decreasing order of value, $x_i$ is the least-valued item in the resulting bundle $X_j$.
	Moreover, the condition $D'[m-i,\mathbf z]=\top$ certifies that the remaining $m-i$ items admit a completion in which, for every other agent, the bundle-value difference toward $a_j$ is at most $x_i$, so EFx is preserved.
	Hence every transition produces a partial allocation that is EFx-extendable.
	
	Conversely, let $\mathcal{X}$ be an EFx-extendable allocation of the first $i$ items (ordered by decreasing value). In $\mathcal{X}$, the item $x_i$ is assigned to some agent $a_j$ and removing $x_i$ still results in an EFx-extendable allocation of the first $i-1$ items. The feasibility of completing $\mathcal{X}$ with the remaining items is captured by the entry $D'[m-i,\mathbf z]=\top$ for the corresponding state $\mathbf z$.
	Therefore, every EFx-extendable allocation of the first $i$ items is represented by some transition.
	Finally, since $D[i,\mathbf y]$ takes the minimum over all transitions from feasible states, it stores the minimum total cost among all EFx-extendable partial allocations of the first $i$ items with value vector $\mathbf y$.
\end{proof}

\begin{proof}[Proof of \Cref{thm:polybound-cost}]
	Construct $D'$ and $D$ in  $\mathcal{O}(n(m+1)\cdot(v(M)+1)^n)$ time as described above.  
	At the end, examine all entries $D[m,\mathbf y]$ with $\sum_j y_j = v(M)$ and return the one with the minimum cost over all such allocations.
	
	By \Cref{thm:dppoly-cost}, $D[m,\mathbf y]$ gives the minimum cost of any EFx allocation with value vector $\mathbf y$. Taking the minimum over all $\mathbf y$ yields exactly the global optimum cost. 
	Since EFx allocations always exist under identical valuations, entries that remain $\infty$ because they cannot be extended to an EFx allocation are harmless, as no EFx allocation has worse cost than $\infty$, so these entries never affect the optimum. Hence the algorithm decides \minEfx\ correctly. 
	
	When $n$ is constant and item values are polynomially bounded my $m$, the algorithm runs in polynomial time.
\end{proof}

\subsection{Hardness of Approximation}

\begin{theorem}\label{thm:inapprox2}
	For any $\rho > 1$, unless $P = N P$ \minEfx \ is not approximable in polynomial time within a factor of $\rho$.
\end{theorem}

We prove this theorem, by giving a gap-preserving reduction from the following NP-hard (\Cref{thm:np-special} $\star$) problem:

\newcommand{\specialparti}{\textsc{Restricted weight equal-cardinality partition}}

\problemdef{\specialparti}{sp}
{A multiset $S = \{s_1, \dots, s_m\}$ of $m$ positive integers with $\sum_{s \in S} s = 2\cdot T$, such that each integer $s \in S$ is smaller than $\frac{2T}{m} + \epsilon$, for some $\epsilon > 0$.}
{Is there a partition of $S$ into $S_1,S_2$ such that every subset sums up to $T$ and $|S_1|=|S_2|$?}

\appendixproof{thm:np-special}{
	\begin{lemma}\label{thm:np-special}
		\specialparti \ is NP-hard.
	\end{lemma} 
	In the following proof we give a polynomial-time many-one reduction from the weakly NP-hard \citep{NP-co1} problem \equalparti \ (\Cref{prob:equalparti}).

	\problemdef{\equalparti}{ep}
	{A multiset $S = \{s_1, \dots, s_m\}$ of $m$ positive integers with $\sum_{s \in S} s = 2\cdot T$.}
	{Is there a partition of $S$ into $S_1,S_2$ such that every subset sums up to $T$ and $|S_1|=|S_2|$?}
	
	\begin{proof}[Proof (\Cref{thm:np-special})]
		Let $I = S$ be an instance of \equalparti, where $\sum_{s \in S} s = 2\cdot T$ and $|S| = m$. We construct an instance $I' := S'$ of \specialparti \ in polynomial time as follows.
		
		\smallskip\noindent
		\textbf{Construction:}
		For each integer $s_i \in S$ we add an integer $s'_i$ to $S'$ such that $s_i' = s_i + mT$. Note that the sum of elements in $S'$ is
		\[
		\sum_{s' \in S'} s' = \sum_{s \in S} \left( s + mT \right) = \sum_{s \in S} s + \sum_{s \in S} mT = 2T + m^2T = (m^2+2)T = 2T'
		\]
		
		This matches the required condition for the input to \specialparti.
		Let  $\epsilon = T$.  This choice of $\epsilon$ ensures that it is small enough relative to the integers in  $S'$.
		Each integer $s_i' \in S'$  will be smaller than $\frac{2T'}{m} + \epsilon$ since
		\[
		s_i' = s_i + mT \leq (m+1)T \leq \frac{m^2T}{m} + T \leq \frac{2T'}{m} + T = \frac{2T'}{m} + \epsilon 
		\]
		This ends the construction.
		
		\smallskip\noindent
		\textbf{Correctness:} We claim that $I$ is a \texttt{yes}-instance if and only if $I'$ is a \texttt{yes}-instance.
		
		\smallskip($\Rightarrow$) Let ${S_1, S_2}$ be a solution of $I$. We claim that $S' = S'_1, S'_2$ is a solution to $I'$ where $S'_i = \bigcup_{s \in S_i} s'$. Note that as $|S_1| = |S_2|$, $|S'_1| = |S'_2|$ and by construction we have that 
		\begin{align*}
			\sum S'_1 = \sum S_1 + \frac{m^2T}{2} = \sum S_2 + \frac{m^2T}{2} = \sum S'_2
		\end{align*} Which means $\sum S'_1 = \sum S'_2$. Hence $S'$ is a solution to $I'$.
		
		\smallskip($\Leftarrow$) Let ${S'_1, S'_2}$ be a solution of $I'$. We claim that $S = S_1, S_2$ is a solution to $I$ where $S_i = \bigcup_{s \in S'_i} s'$. Note that as $|S'_1| = |S'_2|$, $|S_1| = |S_2|$ and by construction we have that 
		\begin{align*}
			\sum S_1 = \sum S'_1 - \frac{m^2T}{2} = \sum S'_2 - \frac{m^2T}{2} = \sum S_2
		\end{align*} We follow that $\sum S_1  = \sum S_2$, meaning $I$ is a \texttt{yes}-instance for \equalparti.
		\smallskip
		
		It follows that \specialparti \ is NP-hard.
\end{proof} }

%Given \Cref{thm:np-special}, we can now prove \Cref{thm:inapprox2} as follows:
\begin{proof}[Proof (\Cref{thm:inapprox2})]
	Let $I = S$ be an instance of \specialparti, where $\sum_{s \in S} s = 2T$ and $|S| = m$. We construct an instance of \minEfx \ in polynomial time as follows.
	
	\smallskip\noindent
	\textbf{Construction.}  
	For each integer $s_i \in S$ we create an item $x_i \in M$ with $v(x_i) = s_i$.  
	Additionally, we add two items $\{x_{m+1}, x_{m+2}\}$ with $v(x_{m+1}) = v(x_{m+2}) = 0$, and one special item $x^*$ with $v(x^*) = T$.  
	
	The cost functions are defined as:
	\begin{align*}
		c_1(x_i) &= 0 & \forall x_i \in M \setminus \{x^*, x_{m+2}\}, 
		&\qquad c_1(x^*) = c_1(x_{m+2}) = \rho, \\
		c_2(x_i) &= 0 & \forall x_i \in M \setminus \{x^*, x_{m+1}\}, 
		&\qquad c_2(x^*) = c_2(x_{m+1}) = \rho, \\
		c_3(x_i) &= \rho & \forall x_i \in M \setminus \{x^*\}, 
		&\qquad c_3(x^*) = 1.  
	\end{align*}
	Finally, we set $k=1$. This completes the construction. 
	
	\medskip\noindent
	\textbf{Completeness.}  Suppose $S$ is a \texttt{yes}-instance of \specialparti, i.e., $S$ can be partitioned into $S_1$ and $S_2$ with $\sum_{s \in S_1} s = \sum_{s \in S_2} s = T$.  
	Consider the allocation $\mathcal{X} = (X_1, X_2, X_3)$ defined by
	\begin{align*}
		X_1 &= \{x_i : s_i \in S_1\} \cup \{x_{m+1}\},\\
		X_2 &= \{x_i : s_i \in S_2\} \cup \{x_{m+2}\},\\
		X_3 &= \{x^*\}.
	\end{align*}
	By construction, $v(X_1) = v(X_2) = v(X_3) = T$, so $\mathcal{X}$ is envy-free and thus also EFx.  
	The cost is
	\[
	\text{cost}(\mathcal{X}) = c_1(X_1) + c_2(X_2) + c_3(X_3) = 0 + 0 + 1 = 1.
	\]
	Hence there exists an EFx allocation of cost at most $1$.
	
	\smallskip\noindent
	\textbf{Soundness:}
	We have $\sum_{x \in M} v(x) = 3T$, and for any allocation $\mathcal{X} = (X_1, X_2, X_3)$ it holds that
	\[
	\text{cost}(\mathcal{X}) = c_1(X_1) + c_2(X_2) + c_3(X_3).
	\]
	Now suppose $S$ is a \texttt{no}-instance of \specialparti. Since $S$ cannot be partitioned evenly, at least one bundle must have value strictly larger than $T$, and another strictly smaller.  
	If $v(X_3) > T$, then $a_3$ must hold at least two items, which forces $c_3(X_3) \geq \rho$. If instead $v(X_1) > T$, EFx requires $v(X_1) - v(\min_v X_1) \leq T$, since some other bundle has value below $T$. This condition excludes $x_{m+1}$ from $X_1$, which in turn forces either $c_2(X_2) \geq \rho$ or $c_3(X_3) \geq \rho$. By symmetry, the same reasoning applies to $X_2$ and $x_{m+2}$.  
	Thus every EFx allocation must incur $\text{cost}(\mathcal{X}) \geq \rho$.
	
	\noindent Consequently, the ratio between the \texttt{yes}- and \texttt{no}-cases is $\rho/1 = \rho$. Therefore, unless $P = NP$, \minEfx \ is not approximable within a factor of $\rho$ in polynomial time.
\end{proof}

\begin{corollary}\label{thm:inapprox3}
	Unless $P = NP$, $\minEfx \notin \mathrm{APX}$.
\end{corollary}

\section{Restricted cost functions}\label{chap:npmin}

\smallskip Building on the above results, we now restrict our attention to a simplified setting in order to analyze the algorithmic complexity more systematically. In this variant, costs are assumed to be \emph{item-independent}: the cost for agent $a_j$ depends only on a multiplicative factor $\alpha_j$ applied to the value of the bundle $X$, i.e.,
\begin{align*}
	c_j(X) = \alpha_j \cdot v(X).
\end{align*}

With a slight abuse of notation, we will henceforth represent the input for the cost functions as a vector, where each coordinate corresponds to the cost factor of an agent.

\smallskip
To illustrate the problem, consider a disaster relief operation where resources must be distributed among several affected regions. Let $M$ denote the available supplies and $A$ the set of regions. Each item has a value $v(x)$, and each region $a_i$ has a cost factor $\alpha_j$ capturing delivery difficulty (e.g., distance or accessibility). The goal is to allocate resources so that the outcome is fair (EFx) while minimizing the total delivery cost $\sum_{a_j \in A} \alpha_j \cdot v(X_j)$ (see \Cref{example} in the appendix).

\toappendix{
	
	In the context of EFx, we can represent an allocation \( X \) as a partition of stacks, by ordering the items within each subset (\Cref{value}). 
	
	\begin{definition}(\textbf{Stack})\label{stack}
		A \emph{stack} is a set of items arranged in non-increasing order of their values. Formally, it is an ordered set \( S \subseteq M \) such that \( v(x_i) \geq v(x_{i+1}) \) for all \( i \in \{1, \dots, |S| - 1\} \), where \( v(x) \) denotes the value of item \( x \).
	\end{definition}
	
	We define the \emph{minStack} as the minimum value of any stack given an EFx allocation and \emph{maxStack} as the maximum value of any stack. Note that \( v(\text{maxStack}) - v(x_{|\text{maxStack}|}) \leq v(\text{minStack}) \). This implies that the highest total value of a stack to which we can still add a smaller-valued item without violating the EFx condition is bounded by \text{minStack}, as illustrated in \Cref{minStack}. Organizing items in a stack according to their values helps us illustrate and analyze the EFx property more clearly by allowing us to easily identify the item with the smallest value, which is always at the top.
	
	\begin{figure}[h!]
		\centering
		\begin{minipage}[b]{0.4\textwidth}
			\centering
			\begin{tikzpicture}[scale=1]
				% Second image
				\foreach \x in {0, 3} {
					\foreach \y in {0} {
						\draw[fill=airforceblue] (\x, \y) rectangle (\x+1, \y+1.5);
					}
				}
				\draw[ultra thick,airforceblue,decorate,decoration={brace,amplitude=12pt}]  (-.2, 0) -- (-.2, 1.5);
				\node[airforceblue] at (-1.25, .75){$v(x_1)$};
				
				\foreach \x in {0} {
					\foreach \y in {1.5} {
						\draw[fill=airforceblue] (\x, \y) rectangle (\x+1, \y+1.5);
						\draw[ultra thick,airforceblue,decorate,decoration={brace,amplitude=12pt}]  (-.2, \y) -- (-.2, \y+1.5);
						\node[airforceblue] at (-1.25, \y +.75){$v(x_2)$};
					}
				}
				\foreach \x in {3} {
					\foreach \y in {1.5, 2.5} {
						\draw[fill=aliceblue] (\x, \y) rectangle (\x+1, \y+1);
						\draw[ultra thick,aliceblue,decorate,decoration={brace,amplitude=12pt}]  (\x - 0.2, \y) -- (\x - 0.2, \y+1);
					}
				} 
				\draw[ultra thick,airforceblue,decorate,decoration={brace,amplitude=12pt}]  (2.8, 0) -- (2.8, 1.5);
				\node[airforceblue] at (1.75, .75){$v(x_3)$};
				\node[aliceblue] at (1.75, 2){$v(x_4)$};
				\node[aliceblue] at (1.75, 3){$v(x_5)$};
				
				\node[black] at (0.5, 2.25){$x_2$}; 
				\node[black] at (0.5, 0.75){$x_1$}; 
				\node[black] at (3.5, 0.75){$x_3$}; 
				\node[black] at (3.5, 2){$x_4$}; 
				\node[black] at (3.5, 3){$x_5$}; 
				
			\end{tikzpicture}
			\caption{Representation of an EFx allocation using stacks}\label{value}
		\end{minipage}
		\hspace{2em}
		\begin{minipage}[b]{0.4\textwidth}
			\centering
			\begin{tikzpicture}[scale=1]
				% Second image
				\foreach \x in {0, 3.5} {
					\foreach \y in {0} {
						\draw[fill=airforceblue] (\x, \y) rectangle (\x+1, \y+1.5);
					}
				}
				\foreach \x in {0} {
					\foreach \y in {1.5, 3} {
						\draw[fill=airforceblue] (\x, \y) rectangle (\x+1, \y+1.5);
					}
				}
				\foreach \x in {3.5} {
					\foreach \y in {1.5, 2.5} {
						\draw[fill=aliceblue] (\x, \y) rectangle (\x+1, \y+1);
					}
				}
				\draw[ultra thick,airforceblue,decorate,decoration={brace,amplitude=12pt}]  (-.2, 0) -- (-.2, 4.5);
				\draw[ultra thick,aliceblue,decorate,decoration={brace,amplitude=12pt}]  (3.3, 0) -- (3.3, 3.5);
				\node[airforceblue] at (-1.7, 2.25){maxStack}; 
				\node[aliceblue] at (2, 1.75){minStack}; 
			\end{tikzpicture}
			\caption{Representation of the maximum and  minimum stackable values}\label{minStack}
		\end{minipage}
		
	\end{figure}	
	
	\begin{example}\label{example}
		\smallskip
		Imagine a disaster relief operation where various resources—such as food, water, medical supplies, and shelter materials—need to be distributed among several affected regions. \( M = \{x_1, \dots, x_m\} \) represents the set of resources, \( A = \{a_1, \dots, a_n\} \) represents the set of regions or communities affected by the disaster, \( v \) represents the valuation function, which indicates the value of each resource, and $(c_i)_{a_i \in A}$ represent the cost functions, which includes the cost of transporting the resources from the central depot to the regions. The goal is to allocate the resources in a way that ensures fairness among the regions (achieving an EFx allocation) while minimizing the overall  transportation cost.
		
		\smallskip
		Suppose there are four regions \(A = \{a_1, a_2, a_3, a_4\}\) and two types of resources - {large care packages}: \( x_1, x_2, x_3, x_4, x_5, x_6, x_7, x_8 \)
		and {small care packages}: \( x_9, x_{10},x_{11} \), both containing essential items, with valuations \(v(\text{Large}) = 15 \) and \(v(\text{Small}) = 10\).
		
		In this case, the cost functions $(c_i)_{a_i \in A}$ could be influenced by the distance from the central depot to each region.  For example, regions situated at a greater distance from the depot may incur higher transportation costs. As illustrated in \Cref{average}, a possible EFx allocation, could provide the four regions with similar items to ensure overall fairness.
		
		However, in such scenarios, financial resources are limited, and the costs associated with transportation increase in proportion to the weight of the load. Therefore, it would be beneficial to reduce the load for more distant regions.
		
		In order to achieve a more efficient allocation, we can define cost functions $c_i(X) = \alpha_i \cdot v(X)$ for $a_i \in A$, where each factor $\alpha$ represents the connectivity of a region. The term "connectivity" in this context refers to how well a region is linked within the network, with a lower factor indicating higher level of connectivity (less cost) and a higher factor indicating lower level of connectivity (more cost). Incorporating this constraint allows us to reduce transportation costs, improve the assignment process, and obtain an EFx assignment (Here let $c(X) = 0 \cdot v(X_1) + 1 \cdot v(X_2) + 1 \cdot v(X_3) + 2 \cdot v(X_4)$).

		\begin{figure}[h!]
			\centering
			\begin{minipage}[b]{0.4\textwidth}
				\centering
				\begin{tikzpicture}[scale=0.8]
					% Second image
					\foreach \x in {0, 1.5, 3, 4.5} {
						\foreach \y in {0, 1.5} {
							\draw[fill=airforceblue] (\x, \y) rectangle (\x+1, \y+1.5);
						}
					}
					\foreach \x in {0, 1.5, 3} {
						\foreach \y in {3} {
							\draw[fill=aliceblue] (\x, \y) rectangle (\x+1, \y+1);
						}
					}
				\end{tikzpicture}
				\caption{An EFx allocation averaging the bundle values}\label{average}
			\end{minipage}
			\hfill
			\begin{minipage}[b]{0.4\textwidth}
				\centering
				\begin{tikzpicture}[scale=0.8]
					% Second image
					\foreach \x in {0, 1.5, 3} {
						\foreach \y in {0, 1.5} {
							\draw[fill=airforceblue] (\x, \y) rectangle (\x+1, \y+1.5);
						}
					}
					%\draw[ultra thick,airforceblue,decorate,decoration={brace,amplitude=12pt}]  (-.2, 0) -- (-.2, 3);
					
					%\node[airforceblue] at (-1, 1.5){$S$}; 
					
					\foreach \x in {0, 1.5} {
						\foreach \y in {3} {
							\draw[fill=airforceblue] (\x, \y) rectangle (\x+1, \y+1.5);
						}
					}
					\foreach \x in {4.5} {
						\foreach \y in {0, 1, 2} {
							\draw[fill=aliceblue] (\x, \y) rectangle (\x+1, \y+1);
						}
					}
				\end{tikzpicture}
				\caption{An EFx Allocation minimizing the cost function}\label{fig:minCost}
			\end{minipage}
			
		\end{figure}
	\end{example}
}
In this context, it is critical to ensure a fair distribution of resources, such that no region feels significantly disadvantaged, particularly in the context of disaster relief. The additional constraint of minimizing transportation costs addresses the practical aspect of limited budgets and resource efficiency.

\subsection{NP-hardness}

The factor model reflects settings where agents agree on the relative value of items but differ in their efficiency or budget.  For instance, distributing identical supplies to two regions may yield equal value,
but delivering to a remote region is uniformly more expensive.  

Furthermore, this restriction enables the design of more efficient approximation algorithms, offering a trade-off between fairness guarantees and computational tractability. However, despite this simplification, the problem remains computationally challenging. In fact, we establish hardness even in the most elementary nontrivial case.

\begin{theorem}[$\star$]\label{thm:np}
	\minEfx \ is NP-hard, even with two agents under restricted cost functions.
\end{theorem}

In the following proof we give a polynomial-time many-one reduction from the following weakly NP-hard \citep{NP-co1} problem:

\newcommand{\equalparti}{\textsc{Equal-cardinality partition}}

\problemdef{\equalparti}{equalparti}
{A multiset $S = \{s_1, \dots, s_m\}$ of $m$ positive integers with $\sum_{s \in S} s = 2\cdot T$.}
{Is there a partition of $S$ into $S_1,S_2$ such that $|S_1| = |S_2|$ and every subset sums up to $T$?}
\begin{proof} Let $I = S$ be an instance of \equalparti.
	If $m$ is odd or some $s_i > T$, then $I$ is a trivial  \texttt{no}-instance of \equalparti{}, and we may map it to any fixed \texttt{no}-instance of \minEfxOptim with restricted cost functions. Hence assume $m$ is even and $s_i \leq T$ for all $i$.
	
	We construct an instance $I' := (A, M, v, c)$ of \minEfx \ with restricted cost functions in polynomial time as follows. 
	\vspace{.5em}\\
	\textbf{Construction:}
	Define 
	\[
	B := (m+1)(2T+1).
	\]
	For each $s_i \in S$, add a number item $x_i$ with valuation $v(x_i) = s_i + B$. Additionally include one item $x_{\min}$ with $v(x_{\min}) = T+1$. Finally we set the cost factors to $\alpha_1=0$ and $\alpha_2=1$. This finishes the construction, which clearly takes polynomial time.
	
	We show that $I$ is a \texttt{yes}-instance if and only if there exists a solution to $I'$ of cost  $\leq \tfrac{m}{2} \cdot B + T$.
	We again refer to the appendix for a proof of correctness.
	\end{proof}
	\appendixproof{thm:np}{
	($\Rightarrow$) Let $(S_1,S_2)$ be a solution of $I$. Consider the allocation
	\[
	X_1 = \{x_i \mid s_i \in S_1\} \cup \{x_{\min}\}, \qquad
	X_2 = \{x_i \mid s_i \in S_2\}.
	\]
	Then $v(X_1) = \tfrac{m}{2}B + 2T + 1$, $v(X_2) = \tfrac{m}{2}B + T$, and therefore $\text{cost}(X) = v(X_2) \leq \tfrac{m}{2} \cdot B + T$. Since $\min_v X_1 = x_{\min}$,
	\[
	v(X_1) - v(\min_v X_1) = \tfrac{m}{2}B + 2T + 1 - (T+1) = \tfrac{m}{2}B + T = v(X_2),
	\]
	so $a_2$ does not strongly envy $a_1$ (and $a_1$ clearly does not envy $a_2$). Hence $X$ is EFx and a valid solution to $I'$.
	
	($\Leftarrow$) Let $\mathcal{X} = (X_1,X_2)$ be a solution of $I'$. Note that
	\[
	v(M) = mB + 3T + 1, \qquad v(X_2) = cost(\mathcal{X}) \leq \tfrac{m}{2}B + T, \tag{1}
	\]
	so
	\[
	v(X_1) = v(M) - v(X_2) \geq \tfrac{m}{2}B + 2T + 1. \tag{2}
	\]
	
	Let $t_j$ denote the number of number-items in $X_j$ (so $t_1 + t_2 = m$), and $S_j =\sum_{x_i \in X_j} s_i$ the sum of their associated values (so $S_1 + S_2 = 2T$). Let $r \in \{0,1\}$ indicate whether $x_{\min} \in X_1$ ($r=1$) or $x_{\min} \in X_2$ ($r=0$). Then
	\[
	v(X_1) = t_1 B + S_1 + r(T+1), \quad v(X_2) = t_2 B + S_2 + (1-r)(T+1). \tag{3}
	\]
	
	By the choice of $B$, inequalities (2) and (3) imply $t_1 = t_2 = \tfrac{m}{2}$, since otherwise one agent would clearly envy the other. Thus $|S_1|=|S_2|$. With $t_2 = \tfrac{m}{2}$ and  inequalities (1) and (2) we get
	\[
	S_2 + (1-r)(T+1) \leq T \quad \Rightarrow \quad S_2 \leq T - (T+1) + r(T+1) = r(T+1)-1. 
	\]
	As all $s_i \geq 0$, we have $S_2 \geq 0$, hence $r = 1$, implying $x_{\min} \in X_1$ and $S_2 \leq T$.
	
	Now in order for $\mathcal{X}$ to be a valid EFx allocation
	\[
	v(X_2) \geq v(X_1) - v(g) \ \text{ for all } g \in X_1.
	\]
	Since $\min_{x\in X_1} v(x) = T+1$, this becomes
	\[
	t_1B + S_1 \leq t_2B + S_2 \quad \Rightarrow \quad S_1 \leq S_2.
	\]
	Together with $S_1 + S_2 = 2T$ and $S_2 \leq T$, this forces $S_1 = S_2 = T$.
	
	Thus $(S_1,S_2)$ is an equal-cardinality partition of $S$, so $I$ is a \texttt{yes}-instance. \qed}

We are using a similar technique to prove:

\begin{theorem}[$\star$]\label{thm:snp}
	\minEfx \  under restricted cost functions is NP-hard, even if all numbers in the input are encoded in unary.
\end{theorem}

In the following proof we give a polynomial-time many-one reduction from the following strongly NP-hard \citep{koutecky_et_al:LIPIcs.ISAAC.2020.18} problem:

\newcommand{\binpacking}{\textsc{Tight Bin packing}}

\problemdef{\binpacking}{bin}
{Finite set $S$ of $m$ items, a size $s(i) \in \mathbb{N}$ for each $i \in S$, with $\sum_{i \in S} s(i) = nB$, a positive integer bin capacity $B$, and a positive integer $n$.}
{Is there a partition of $I$ into disjoint sets $S_{1}, \ldots, S_{n}$ such that the sum of the sizes of the items in each $S_{j}$ is $B$ or less?}
\begin{proof}
	Let $I = (S, B, n)$ be an instance of \binpacking, where $\sum_{i \in S} s(i) = n B$ and $|S| = m$. We construct an instance $I' := (A, M, v, c)$ of  \minEfx \ in polynomial time as follows.
	\vspace{.5em}\\
	\textbf{Construction:} Note that all integers in \( S \) are less than or equal to \( B \), as otherwise, the instance would trivially be a \texttt{no}-instance.
	Let $A$ be a set of $n+1$ agents.
	For each item $i \in S$ we add an item $x_i$ to $M$ with valuation $v(x_i) = s(i)$. Furthermore we add two items $ \{x_{m+1}, x_{m+2}\} $ with valuations $v(x_{m+1}) = v(x_{m+2}) = B$. Lastly we set the cost factors $\alpha_1=0$ and $\alpha_i = 1$ for all other $a_i \in A\setminus\{a_1\}$. This finishes the construction. 
	
	We refer to the appendix for the proof of correctness.\end{proof}
	\appendixproof{thm:snp}{
	\textbf{Correctness:} We claim that $I$ is a \texttt{yes}-instance if and only if $I'$ has a solution of cost $\leq nB$.
	
	($\Rightarrow$) Let \( S = (S_1, \ldots, S_n) \) be a solution of \( I \). We claim that \( \mathcal{X} = (X_1, \ldots, X_{n+1}) \) with \begin{align*}
		X_1 = \{x_{m+1}, x_{m+2}\} \text{ and } X_i = \bigcup_{s_j \in S_{i-1}} \{ x_j \} \text{ for all } i \in \{2, \ldots, n+1\}
	\end{align*} is a solution of cost $n B$ to \( I' \). Note that by construction \( v(X_1) = 2 B \) and \( v(X_i) \leq B \) for all \( i \in \{1, \ldots, n+1\} \). Furthermore, since \( v(M \setminus X_1) = nB \), by the pigeonhole principle \citep{graham1981concrete} we can deduce that \( v(X_i) = B \) for all \( i \in \{2, \ldots, n+1\} \). This means that \begin{align*}
		\text{cost}(A) = 0 \cdot v(X_1) + \sum_{i=2}^{n+1} v(X_i) = n \cdot B 
	\end{align*}
	Lastly, we need to show that \( A \) is an EFx allocation. As \( v(X_i) = v(X_j) = B \) for all \( i, j \in \{2, \ldots, n+1\} \), agents \( a_i \) and \( a_j \) do not envy each other. Furthermore, since \( X_1 \) consists of 2 items of equal value, \( v(X_1) - \min_v X_1 = B \), no agent strongly envies \( a_1 \). Hence, \( \mathcal{X} \) is EFx and has cost $\leq nB$.
	
	\smallskip
	($\Leftarrow$) Let $\mathcal{X} = (X_1, \ldots, X_{n+1})$ be a solution of $I'$ with cost $\leq nB$.
	Note that:
	
	\[
	v(M) = (n+2) \cdot B \text{ and }
	nB \geq 0 \cdot v(X_1) + \sum_{i=2}^{n+1} v(X_i)
	\]
	which implies:
	\begin{align*}
		nB &\geq \sum_{i=2}^{n+1} v(X_i)
	\end{align*}
	Since \( v(M) = \sum_{i=1}^{n+1} v(X_i) =(n+2) \cdot B \), we get:
	\begin{align*}
		(n+2) \cdot B &= v(X_1) + \sum_{i=2}^{n+1} v(X_i)\\
		(n+2) \cdot B&\leq v(X_1) + nB \\
		2B &\leq v(X_1)
	\end{align*}
	
	Since $\mathcal{X}$ is an EFx allocation, we have:
	
	\[
	v(X_1) - \min_v X_1 \leq v(X_i) \quad \forall i \in \{2, \dots, n+1\}.
	\] By construction, the largest valuation of any item in $M$ is $B$. Therefore, we can write:
	\[
	v(X_1) - \min_v X_1 \geq v(X_1) - B = 2B - B = B.
	\]
	This implies that $v(X_i) \geq B$ for all \( i \in \{2, \ldots, n+1\} \). Given that the total valuation of the allocations to agents in $A \setminus \{a_1\}$ satisfies $\sum_{i=2}^{n+1} v(X_i) \leq n \cdot B$, it can be concluded that $v(X_i) = B$ for all \( i \in \{2, \ldots, n+1\} \). Consequently, the minimum item in $X_1$ must have value $B$, implying that $v(X_1) - \min_v X_1 = B$. From this, it can be deduced that $X_1$ is constituted of two items: $x_{m+1}$ and $x_{m+2}$. In the event that the maximum value of the set of items $S$ is equal to $B$, it is possible to exchange the corresponding item with one of the additional items.
	
	Lastly, we need to prove that $S = \{S_1, \ldots, S_n\}$ is a solution to $I$, where:
	
	\[
	S_i = \bigcup_{x_j \in X_{i+1}} s_j \quad \text{for all } i \in \{1, \ldots, n\}.
	\]
	Assume that $\sum S_i > B$ for some $i\in [n]$. Given that $\sum_{i = 2}^{n+1} v(X_i) = n \cdot B$, this inequality implies that there exists some $S_j$ such that $\sum S_j < B$. This contradicts the fact that $\sum S_j = v(X_{j+1}) \geq B$, as agent $a_{j+1}$ did not strongly envy agent $a_1$, and we have $v(X_1) - \min_v X_1 = B$. Thus, we conclude that $\sum S_i = B$ for all $i$, meaning that $I$ is a \texttt{yes}-instance for the \binpacking \ problem. \qed
}

\begin{corollary}\label{w1}
	\minEfx\ is $W[1]$-hard when parameterized by the number of agents $n$. 
	%In particular, there is no algorithm running in time $f(n)\cdot |I'|^{O(1)}$ for any computable function $f$ unless $\mathrm{FPT}=\mathrm{W[1]}$.
\end{corollary}

\begin{proof}
	The reduction above maps an instance of \binpacking\ with parameter $k$ (the number of bins) to an instance of \minEfx\ with parameter $n=|A|=k+1$. Since \binpacking\ is $W[1]$-hard parameterized by $k$, the parameter-preserving reduction implies that \minEfx\ is $W[1]$-hard parameterized by $n$.
\end{proof}

\subsection{Parameterized Complexity}

Given the hardness established in the previous section, we now aim to investigate the influence of specific parameters on the computational complexity of \minEfx \ under restricted cost functions, with the objective of identifying conditions under which it becomes more tractable. Specifically we study how the diversity of item types impacts complexity. Fewer item types generally lead to a simpler allocation strategy, thereby decreasing computational demands. This restriction is also natural in many allocation contexts, where items often belong to a small number of well-defined categories.
Recall that under bounded item valuations, \Cref{thm:dppoly-cost} already implies XP in $n$ for bounded valuations.

\begin{theorem}\label{thm:types}
	For instances with restricted cost functions and $\beta$ item types, \minEfx\ can be solved in
	$\mathcal{O}(n^2 \cdot m^{3\beta})$ time.
\end{theorem}

We will describe a simple dynamic programming algorithm.
% inspired by \citet{koutecky_et_al:LIPIcs.ISAAC.2020.18}. \todo{ADD THEOREM!}
\smallskip 

\begin{definition}
	The DP table of an instance $I := (A, M, v, c)$ with $m = |M|$, $n=|A|$ and $\beta$ different item types $s_1, \dots, s_\beta$ is an integer table $D$ with entries
	\begin{align*}
		D[j,\mathbf{x}, j',\mathbf{m'}] \in  \mathbb{N}_0 \cup \{\infty\}.
	\end{align*}
	%In this context, an item \( x \) is considered to be of type \( s_i \) if and only if \( v(x) = s_i \) for all \( x \in M \). 
	Each subset $X \subseteq M$ can be represented by a $\beta$-dimensional vector $\mathbf{x} = (x_1,\dots,x_\beta)$, where $x_i$ denotes the number of items of type $s_i$ in $X$. Let $\mathbf{m} = (m_1,\dots,m_\beta)$ denote the vector representation of the full item set.
	
	The intended meaning of an entry $D[j,\mathbf{x},j',\mathbf{m'}]=c$ is as follows.
	Fix agent $a_j$ to receive the bundle encoded by $\mathbf{x}$.
	Consider the sub-instance consisting of the first $j'$ agents (ordered by non-increasing cost and excluding $a_j$),
	and an item multiset encoded by $\mathbf{m'} \le \mathbf{m}-\mathbf{x}$.
	Then $c$ is the minimum total cost of a feasible partial EFx allocation of these $j'$ agents using the items $\mathbf{m'}$
	(together with $a_j$ holding $\mathbf{x}$).
	If no such allocation exists, we set $D[j,\mathbf{x},j',\mathbf{m'}]=\infty$.
	
		The DP table $D$ is indexed by
		\begin{itemize}
			\item two $\beta$-dimensional integer vectors $\mathbf{x}, \mathbf{m'} \in \{0,\dots,m\}^\beta$,
			where $\mathbf{x}=(x_1,\dots,x_\beta)$ encodes the bundle assigned to a designated agent
			and $\mathbf{m'}=(m'_1,\dots,m'_\beta)$ encodes a subset of the remaining items. Both vectors are
			bounded componentwise by $\mathbf{m}$.
			\item two indices $j,j' \in \{1,\dots,n\}$, where $j$ denotes the designated agent receiving bundle $\mathbf{x}$ and $j'$ denotes the number of agents considered (in non-increasing cost order) in the partial allocation of the items encoded by $\mathbf{m}'$.
		\end{itemize}
	
\end{definition}

\begin{lemma}[$\star$]\label{thm:dpbeta}
	Given an instance $I = (A, M, v, c)$ the dynamic programming table of $I$ can be computed in $\mathcal{O}(n^2\cdot m^{3\beta})$ time.
\end{lemma}

\appendixproof{thm:dpbeta}{

\begin{proof} 
	We begin by defining feasible vectors, which will be useful for constructing the DP table.
	
	\begin{definition}
		A vector $\mathbf{m'} \in \{0,\dots, m\}^\beta$ is said to \emph{satisfy the constraints} if, given a minimum threshold $bound$, it meets the following conditions:
		\begin{align*}
			\sum_{i \in [\beta]} m'_i \cdot s_i - \min_{\substack{i' \in [\beta] \\ m'_{i'} > 0}} s_{i'} &\leq bound, \\
			\sum_{i \in [\beta]} m'_i \cdot s_i &\geq bound.
		\end{align*}
		Here, $bound$ represents the minimum value that a set must achieve, as well as the largest permissible sum for a set excluding its smallest item, in order to satisfy the Efx constraint.
	\end{definition}
	
	The idea of our algorithm is to guess the minimum value each set must attain, and then identify a \minEfx\ where this minimum threshold is met. The number of feasible allocations is bounded by $m^\beta$. To prevent significant envy toward any particular set, the values of all other sets must be at least equal to this minimum value, as visualized in \Cref{minBound}. Furthermore, since this represents the minimum bound, no other set's value can exceed this minimum without including the smallest item. As this set does not necessarily need to have the lowest value, we compute the values for all possible allocations of this set to agents, which amount to $n$ possibilities.
	\begin{figure}[h!]
		\centering
		\begin{minipage}{0.8\linewidth}
			\centering
			\begin{tikzpicture}[scale=0.8]
				% Second image
				\foreach \x in {2, 4, 6} {
					\foreach \y in {0} {
						\draw[fill=airforceblue] (\x, \y) rectangle (\x+1, \y+1.5);
					}
				}
				
				\foreach \x in {2, 6} {
					\foreach \y in {1.5} {
						\draw[fill=airforceblue] (\x, \y) rectangle (\x+1, \y+1.5);
					}
				}
				\foreach \x in {4} {
					\foreach \y in {1.5, 2.5} {
						\draw[fill=airforceblue] (\x, \y) rectangle (\x+1, \y+1);
					}
				}
				
				\draw[fill=airforceblue] (6, 3) rectangle (7, 4.5);
				\draw[ultra thick, dashed,aliceblue] (-1,3)--(8,3);
				
				\draw[fill=aliceblue] (0, 0) rectangle (1, 2);
				\draw[fill=aliceblue] (0, 2) rectangle (1, 3);
				\draw[fill=aliceblue] (0, 3) rectangle (1, 4);
				
				\draw[ultra thick,airforceblue,decorate,decoration={brace,amplitude=12pt}]  (-.2, 0) -- (-.2, 3);
				\node[airforceblue] at (-2.25, 1.5){minBound}; 
				
			\end{tikzpicture}
			\caption{Example illustrating the minimum threshold each set must meet in a feasible EFx allocation}\label{minBound}
		\end{minipage}
	\end{figure}
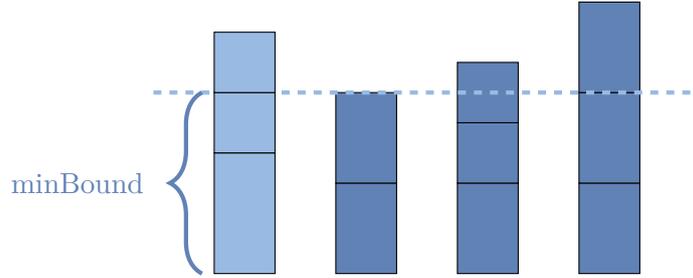
	
	We can now outline how to fill the table and compute its entries.
	
	\noindent
	\textbf{Construction:} For each possible value vector of \( X \), denoted as \( \mathbf{x} \), and each cost factor \( \alpha_j \) with \( a_j \in A \), compute the following:
	\begin{align*}
	bound_{\mathbf{x}} &= \sum_{i \in [\beta]} x_i \cdot s_i 
	%- \min_{\substack{i' \in [\beta] \\ x_{i'} > 0}} s_{i'}
	\\
	\text{cost($j,\mathbf{x}$)} &= \alpha_j \cdot \sum_{i \in [\beta]} x_i \cdot s_i
	\end{align*}
	Sort the remaining agents $A\setminus\{a_j\}$ in decreasing order of their cost factors, and index them as  $j'\in [n-1]$ according to this order. We begin by initializing the table \( D \) with all entries set to $\infty$. Then, set all  $D[j,\mathbf{x},0, \mathbf{0}] =  \text{cost}(j,\mathbf{x})$. 
For each \( j' \in [n-1] \) and every \( 0 \leq \mathbf{m'} \leq \mathbf{m - x} \), if
\(
D[j, \mathbf{x}, j'-1, \mathbf{m'} - \mathbf{z}] \neq \infty,
\)
then update \( D[j, \mathbf{x}, j', \mathbf{m'}] \) whenever
\[
D[j, \mathbf{x}, j', \mathbf{m'}] 
> D[j, \mathbf{x}, j'-1, \mathbf{m'} - \mathbf{z}] + \text{cost}(j, \mathbf{x}),
\]
and the vector \( \mathbf{z} \) satisfies the specified constraints 
(given $bound_\mathbf{x}$). 
	
In other words, for each $j' \in [n-1]$, construct the set $C_{j'}$ of feasible configurations of $X_{j'}$ that satisfy the given constraints.
Then, for each $\mathbf{m'}$ such that $D[j, \mathbf{x}, j'-1, \mathbf{m'}] \neq \infty$,
retain only the configurations that yield the minimum total cost.

Finally, return the minimum value $D[j, \mathbf{x}, n-1, \mathbf{m - x}]$ over all $a_j \in A$ and $\mathbf{x} \in [m]^{\beta}$.
	
	\smallskip\noindent
	\textbf{Running time:}
	The algorithm considers all possible allocations of a single set, resulting in a total of \( n \cdot m^\beta \) allocations. For each of these allocations, the algorithm computes the corresponding EFx bound. In each iteration, the algorithm processes every possible configuration \( \mathbf{m'} \leq \mathbf{m} - \mathbf{x}\), of which there are at most \( m^\beta \) many. For each configuration \( \mathbf{m'} \), it attempts to add each element from the set \( C_{j'} \), which also contains up to \( m^\beta \) feasible configurations. Therefore, the algorithm performs \(  (n \cdot m^\beta) \cdot (n \cdot m^\beta \cdot m^\beta) = n^2 \cdot m^{3\beta}  \) steps in total, leading to a runtime of \( \mathcal{O}(n^2 \cdot m^{3\beta}) \).
	
	\smallskip\noindent
	\textbf{Correctness:}
	Let $\pi: [n-1] \to A \setminus \{a_j\}$ denote the permutation ordering of all agents except $a_j$ in decreasing cost factor. 
	We claim that $D[j, \mathbf{x}, i, \mathbf{m'}] = c$ if and only if there exists an EFx-feasible assignment 
	$\mathcal{X}=(X_{\pi(1)}, \dots, X_{\pi(i)}, X_j)$ 
	such that $X_j$ corresponds to the value vector $\mathbf{x}$ and $cost(\mathcal{X})=c$. 
	We prove the claim by induction on $i$.
	
	Initially, for \( i=0 \), all entries of 
$D$ are trivially valid, since only agent $a_j$ is considered, and thus no pairwise comparisons are possible and strong envy cannot arise.
	
	Now assume that for an agent \( a_i \in A \), \( D[j,\mathbf{x}, i,\mathbf{m'}] \) is set to $c$ if and only if there exists an allocation $(X_{\pi(1)}, \dots, X_{\pi(i)}, X_j)$, where $(X_{\pi(1)}, \dots, X_{\pi(i)})$ is a partition of the good set represented by vector $\mathbf{m}'$ and the bundle \( X_j \) can be represented by \( \mathbf{x} \), that is EFx feasible, and has minimal cost $c$. We want to show that this also holds for $i+1$.
	
	\smallskip
	($\Rightarrow$) Suppose \( D[j,\mathbf{x}, i+1,\mathbf{y}] \) is set to $c'$. By our induction assumption, this implies that there exists a table entry \( D[j,\mathbf{x}, i,\mathbf{y'}] \) that is set to $c' - \sum_{z_j \in \mathbf{z}} \alpha_{i+1} \cdot z_j \cdot v(x_j)$, where \( y_k = y'_k + z_k \) for all \( k \in [\beta] \), with \( \mathbf{z} \) being a vector satisfying the constraints and \( y_k \leq m_k - x_k \).
	
	Given that there is an EFx feasible assignment $\mathcal{X} = (X_{\pi(1)}, \dots, X_{\pi(i)}, X_j)$ such that the value of the bundle \( X_j \) is represented by \(\mathbf{x}\) and the sum of the items of each type considered so far is \( y'_k \) for all \( k \in [\beta] \), we can construct an assignment \( \mathcal{X}' \) where the value of the bundle \( X'_{i+1} \) is represented by the vector \( \mathbf{z} \), while \( X'_{i'} = X_{i'} \) for all other \( i' \in \{1, \dots, i\} \). Since \( \mathbf{z} \) satisfies the constraints, no strong envy is created with this allocation. By construction, 
	\[
	\sum_{i \in [\beta]} z_i \cdot s_i = v(X'_{i+1}) \geq bound_{\mathbf{x}}
	\]
	and
	\[
	\sum_{i \in [\beta]} z_i \cdot s_i - \min_{\substack{i' \in [\beta] \\ z_{i'} > 0}} s_{i'} = v(X'_{i+1}) - \min_v X'_{i+1} \leq bound_{\mathbf{x}},
	\]
	which ensures that the assignment remains EFx. Furthermore, since all other bundles remain unchanged, EFx feasibility continues to hold for \( \mathcal{X}' \).
	
	Finally, we update the cost of the current allocation if it is uninitialized or if the new cost is lower than the previously set cost. The agents are sorted in decreasing order of their cost factor, so minimizing the cost guarantees that the current allocation is the most efficient and least expensive option available. Thus, we conclude that there exists a valid assignment $\mathcal{X}=(X_{\pi(1)}, \dots, X_{\pi(i)}, X_j)$, in which the value of bundle \( X_j \) can be represented by \( \mathbf{x} \), such that the cost of \( \mathcal{X} \) is minimal.
	
	\smallskip
	($\Leftarrow$) Let $\mathcal{X}=(X_{\pi(1)}, \dots, X_{\pi(i+1)}, X_j)$ be an EFx assignment with minimal cost $c'$, where the value of \( X_{\pi(i+1)} \) is represented by the vector \( \mathbf{z} \) and the value of \( X_j \) is represented by the vector \( \mathbf{x} \). Then, there exists an assignment $\mathcal{X}=(X_{\pi(1)}, \dots, X_{\pi(i)}, X_j)$ that is EFx feasible. 
	%Since adding an agent can only increase the cost, the cost of this assignment is also minimal. 
	Let \( \mathbf{y'} \) be the value vector of the items assigned in this allocation. 
	By the inductive hypothesis, \( D[j, \mathbf{x}, i, \mathbf{y'}] \) stores the minimal achievable cost for this configuration, 
	that is, \( c = D[j, \mathbf{x}, i, \mathbf{y'}] \). 
	For every feasible extension vector \( \mathbf{z} \) satisfying the EFx constraints, 
	the resulting allocation extends \( \mathbf{y'} \) to a new value vector 
	\( \mathbf{y} = \mathbf{y'} + \mathbf{z} \) with corresponding cost update. 
	Since all feasible \( \mathbf{z} \) are considered, 
	\( D[j, \mathbf{x}, i+1, \mathbf{y}] \) stores the minimal achievable cost among all valid extensions, 
	thereby preserving the inductive invariant and completing the induction step.
	
	\smallskip
	We have shown that if \( D[j, \mathbf{x}, i, \mathbf{y}] \) is finite for some \( i \in [n-1] \), 
	then there exists an assignment of a subset of items in \( M \) to the first \( i \) agents, 
	sorted by decreasing cost factor, 
	such that the total value vector is \( \mathbf{y} \), 
	the allocation satisfies EFx feasibility, 
	and \( D[j, \mathbf{x}, i, \mathbf{y}] \) represents the minimal achievable cost for this configuration. 
	Moreover, this property propagates to \( D[j, \mathbf{x}, i+1, \mathbf{y}] \) through the update step, 
	thereby preserving both EFx feasibility and cost minimality.
	
	\smallskip
	$\therefore$ By the principle of induction, the claim holds for all $i\in [m]$.
\end{proof}

Given \Cref{thm:dpbeta}, we can prove \Cref{thm:types} as follows:
\begin{proof}[Proof (\Cref{thm:types})]
	Let  $I = (A, M, v, c)$ be an instance of \minEfx \ with cost factors and identical valuation functions $v: M \rightarrow \NN$ over a set of items $M = \{x_1, \dots, x_m\}$. We start by sorting the agents in decreasing order according to their cost factor, such that  $c_1 \geq c_2 \geq \dots \geq c_n$. Then we construct the table $D$ in $\mathcal{O}(n^2 \cdot m^{3\beta})$ time. For all $a_j \in A$ and $\mathbf{x} \in m^\beta$ such that \( D[j,\mathbf{x}, n-1,\mathbf{m-x}]\) we return the entry with the smallest cost.
	
	\smallskip\noindent
	\textbf{Correctness:} Let \( \mathcal{X} = (X_1, \dots, X_n) \) be an allocation of items that satisfies the EFx condition. Each agent's allocation can be described by a type vector \(\mathbf{x}_i\), which indicates the quantity of different types of items in their set. Consider the set \( X_j \), represented by the vector \(\mathbf{x}\), which has the highest value among all agents after removing their smallest items. This set \( X_j \) sets a benchmark for what we consider the minimum acceptable value that each set must attain to maintain fairness. As the dynamic programming entry \( D[j, \mathbf{x}, n-1, \mathbf{m-x}] \) is defined to be \texttt{true} if there exists a valid EFx allocation for the other \( n-1 \) agents that respects the given vector configuration \(\mathbf{x}\) to $a_j$, and the whole item set \(\mathbf{m}\), our algorithm would have set this value to \texttt{true}. Since the agents are sorted by decreasing cost factor, minimizing the overall cost of the allocation also involves minimizing the value of the most expensive set while ensuring that the EFx condition is not violated. This guarantees that the dynamic programming entry \( D[j, \mathbf{x}, n-1, \mathbf{m-x}] \) will reflect the minimal possible cost for a valid EFx allocation at each step. Therefore, if there exists an EFx allocation with a total cost less than or equal to the given threshold, the algorithm will return \texttt{true}. 
\end{proof}
}

\subsection{Approximation}

Even under this restriction on the cost function, we can prove that a nontrivial lower bound remains:
%Building on the inapproximability of \minEfx\ under general cost functions, we 
%now show that the situation improves significantly when restricting the class of cost functions.
%Despite this restriction, 
%show that a nontrivial lower bound remains:

\begin{theorem}\label{thm:inapprox}
	For any $\epsilon > 0$, unless $P = N P$ \minEfx\ under restricted cost functions is not approximable in polynomial time within a factor of $\frac{4}{3} -\epsilon$.
\end{theorem}

We prove this via a gap-preserving reduction from \specialparti\ (see \Cref{thm:np-special}).

\begin{proof}[Proof (\Cref{thm:inapprox})]
	Let $I = S$ be an instance of \specialparti, where $\sum_{s \in S} s = 2T$ and $|S| = m$. We construct an instance of \minEfx \ in polynomial time as follows.
	
	\smallskip\noindent
	\textbf{Construction:} 
	For each integer $s_i \in S$ we add an item $x_i$ to $M$ with $v(x_i) = s_i$. Furthermore we add two items $ \{x_{m+1}, x_{m+2}\} $ with valuations $v(x_{m+1}) = v(x_{m+2}) = T$.
	
	\smallskip\noindent\textbf{Completeness:} Let $S= {S_1, S_2}$ be a \texttt{yes}-instance for \specialparti. We claim that $\mathcal{X} = (X_1, X_2, X_3)$ with \begin{align*}
		X_1 = \{x_{m+1}, x_{m+2}\} \quad X_2 = \bigcup_{s_i \in S_1}\{ x_i \}  \text{ and } X_3 = \bigcup_{s_i \in S_2} \{x_i \}
	\end{align*} is a \minEfx. Note that by construction $v(X_1) = 2T$ and $v(X_2) = v(X_3) = T$, which means \begin{align*}\text{cost}(\mathcal{X}) = 0 \cdot v(X_1) + v(X_2) + v(X_3) = T + T = 2 T\end{align*} Lastly we have to show that $\mathcal{X}$ is an EFx allocation. As $v(X_2) = v(X_3) = T$, $a_2$ and $a_3$ do not envy each other. Since $X_1$ consists of $2$ items of equal value, \begin{align*}
		v(X_1) - \min_v X_1 = T \leq v(X_2) = v(X_3)
	\end{align*} meaning no one strongly envies $a_1$. Hence there is an EFx allocation of cost $\leq 2T$.
	
	\smallskip\noindent
	\textbf{Soundness:} Let $S$ be a \texttt{no}-instance for \specialparti. Note that $\sum_{x \in M} v(x) = 4 T$ and \begin{align*}\text{cost}(\mathcal{X}) = 0 \cdot v(X_1) + v(X_2) + v(X_3) =  v(X_2) + v(X_3) = 4T - v(X_1)\end{align*}
	% Note: the largest set will have a value of at least $\frac{4T}{3}$
	Since the items in $S$ can not be equally split, we can deduce that no bundle can contain both $\{x_{m+1}, x_{m+2}\}$. By definition the next largest item has value at most $\frac{2T}{m} +\epsilon'$. From \Cref{thm:diff} we know that this corresponds to the maximum difference between $2$ sets. We can follow that $v(X_1) \leq \frac{4T}{3} + \frac{2T}{m} + \epsilon'$, meaning \begin{align*}\text{cost} \geq 4T - v(X_1) = 4T - \left(\frac{4T}{3} + \frac{2T}{m} + \epsilon'\right) = \frac{8T}{3} - \frac{2T}{m} - \epsilon'\end{align*}
	
	\noindent Following these results, we have \begin{align*}\frac{\frac{8T}{3} - (\frac{2T}{m} +\epsilon')}{2T} = \frac{4}{3} - (\frac{1}{2m} +\epsilon')\end{align*} and as $m$ increases, $\epsilon = \frac{1}{2m} + \epsilon'$ becomes very small.
	Therefore, unless $P = NP$, \minEfx \ with cost factors is not approximable within a factor of $\frac{4}{3}$ in polynomial time.
\end{proof}

\section{Discussion}

We introduced the \minEfx{} problem and analyzed its algorithmic complexity under both general and restricted cost functions. Our results leave several natural directions for further work.

We focused on two representative cost models, but many others remain unexplored. For instance, cost structures inspired by generalized bin packing \citep{GCBP} or scheduling \citep{sched} could lead to new approximation techniques or hardness barriers. A systematic classification of cost models according to the tractability of \minEfx{} would provide a clearer picture of the landscape.

Our work concentrated on identical valuations, since EFx existence results are currently limited. However, recent progress for structured settings, such as a few types of agents \citep{MAHARA2023115,ghosal2025almostfullefxand} or bi-valued instances \citep{AMANATIDIS202169}, suggests that our framework could be extended to these domains. Analyzing \minEfx{} under such restricted valuation classes would bridge the gap between the idealized identical case and more heterogeneous environments.

Most of our techniques extend to EF1 with only minor modifications (essentially making the “forced” item the maximum-valued one). Exploring EF1 with differing valuations or other fairness benchmarks such as maximin share (MMS), or maximizing Nash social welfare subject to cost minimization, could yield richer trade-offs between fairness and budget feasibility.

Finally, the empirical study of \citet{Morozov_Ignatiev_Dementiev_2024} suggests that EFx allocations are scarce in practice. This raises the question of whether efficient algorithms can nevertheless find low-cost EFx allocations quickly in real-world data, or whether relaxations (e.g., EF1 or approximate EFx) are necessary. A systematic experimental evaluation would complement our theoretical results and provide insights into the practical feasibility of \minEfx.

%%%%%%%%%%%%%%%%%%%%%%%%%%%%%%%%%%%%%%%%%%%%%%%%%%%%%%%%%%%%%%%%%%%%%%%%

 	\section*{Acknowledgments}
 	A substantial portion of this work was completed as part of my Master’s thesis at TU Berlin \citep{mincostefx}, under the joint supervision of Mathias Weller (TU Berlin) and Ola Svensson (EPFL). I am grateful to my PhD supervisor Robert Bredereck and to Leon Kellerhals (TU Clausthal) for valuable feedback.
 	I acknowledge support from the Deutsche Forschungsgesellschaft (German Research Foundation, DFG), project COMSOC-MPMS, (grant agreement No. 465371386).
% Add Ola, Robert (and maybe Leon and mathias), as well as DFG funding woop!

% --- Bibliography ---
\bibliographystyle{abbrvnat}
\bibliography{minCostEFx}

\newpage

\appendix
\section{Appendix}

\appendixProofText

\end{document}